%% file: arXiv_v1.tex
\documentclass[journal,12pt,onecolumn,draftclsnofoot]{IEEEtran}
\ifCLASSINFOpdf
\else
\fi
\usepackage[utf8]{inputenc} 
\usepackage[T1]{fontenc}
\usepackage{ifthen}
\usepackage[cmex10]{amsmath} 
\usepackage{enumitem}
\usepackage{amsfonts,amssymb,amsthm,epsfig,epstopdf,url,array}
\usepackage{thmtools,thm-restate}
\usepackage[table]{xcolor}
\usepackage{hyperref}
\usepackage{bbm}
\hypersetup{
    colorlinks=true,
    linkcolor=blue,
    filecolor=magenta,      
    urlcolor=red,
    citecolor=blue,
}
\usepackage{bbm}
\usepackage{times}
\usepackage{algorithm}
\usepackage{balance}
 \usepackage{float}
\usepackage[noend]{algpseudocode}
\theoremstyle{plain}
\newtheorem{theorem}{Theorem}

\usepackage{caption}
\usepackage{subcaption}

\theoremstyle{definition}
\newtheorem{definition}{Definition}
\newcommand\numberthis{\addtocounter{equation}{1}\tag{\theequation}}

\theoremstyle{remark}

\DeclareMathOperator*{\argmax}{arg\,max} 
\DeclareMathOperator*{\E}{\mathbb{E}}

\newcommand{\indep}{\raisebox{0.05em}{\rotatebox[origin=c]{90}{$\models$}}}

\newcolumntype{P}[1]{>{\centering\arraybackslash}p{#1}}

\let\emptyset\varnothing
\allowdisplaybreaks
\begin{document}
%
\title{Algorithms for reconstruction over single and multiple deletion channels\footnote{This work was presented in part at ISIT 2018 \cite{Srini2018} and ISIT 2019 \cite{Srini2019}. This work was supported in part by NSF grants 1705077, 1740047 and
UC-NL grant LFR-18-548554.}}
%
%
%

\author{Sundara~Rajan~Srinivasavaradhan,
        Michelle~Du,
        Suhas~Diggavi and Christina~Fragouli
\thanks{All the authors are with the Department
of Electrical and Computer Engineering, University of California Los Angeles,
CA, 90095, USA. Email: \{sundar, michelleruodu, suhas.diggavi, christina.fragouli\}@ucla.edu.}
}

\maketitle

\begin{abstract}
Recent advances in DNA sequencing technology and DNA storage systems have rekindled the interest in deletion channels. Multiple recent works have looked at variants of sequence reconstruction over a single and over multiple deletion channels, a notoriously difficult problem due to its highly combinatorial nature. Although works in theoretical computer science have provided algorithms which guarantee \textit{perfect reconstruction} with multiple independent observations from the deletion channel, they are only applicable in the large blocklength regime and more restrictively, when the number of observations is also large. Indeed, with only a few observations, perfect reconstruction of the input sequence may not even be possible in most cases. In such situations, maximum likelihood (ML) and maximum aposteriori (MAP) estimates for the deletion channels are natural questions that arise and these have remained open to the best of our knowledge. In this work, we take steps to answer the two aforementioned questions. Specifically: 1. We show that solving for the ML estimate over the single deletion channel (which can be cast as a discrete optimization problem) is equivalent to solving its relaxation, a continuous optimization problem; 2. We exactly compute the symbolwise posterior distributions (under some assumptions on the priors) for both the single as well as multiple deletion channels. As part of our contributions, we also introduce tools to visualize and analyze error events, which we believe could be useful in other related problems concerning deletion channels.
\end{abstract}

\begin{IEEEkeywords}
Deletion channels, Trace reconstruction, symbolwise MAP, Edit graph, Dynamic programming
\end{IEEEkeywords}

%
\IEEEpeerreviewmaketitle

\input{Sections_revised/intro.tex}

\input{Sections_revised/notation.tex}

\input{Sections_revised/single_ML.tex}

\input{Sections_revised/single.tex}

\input{Sections_revised/multiple.tex}

\input{Sections_revised/numerical.tex}
\input{Sections_revised/conclusions.tex}



\bibliographystyle{IEEEtran}
\bibliography{mybib}

\appendix
\input{Sections_revised/Appendix.tex}

\end{document}

%% file: Sections_revised/intro.tex
\section{Introduction}
\label{sec:intro}

Sequence reconstruction over deletion channels, both with and without a codebook, has received considerable attention in the information theory as well as in the theoretical computer science literature. From an information theory perspective, reconstruction over the deletion channel, or more specifically a  maximum-likelihood (ML) argument for the deletion channel, would give further insight on the capacity of the deletion channel, a long-standing open problem (see \cite{mitzenmacher2009survey}). To quote \cite{mitzenmacher2009survey} -- ``at the very least, progress in this direction would likely surpass previous results on the capacity of the deletion channels''. Yet, there are no results on reconstruction over a deletion channel with statistical guarantees. In this work, we take steps in this direction.

In this space, the problem of \textit{trace reconstruction}, as introduced in \cite{Batu2004}, has also received renewed interest in the past few years (see \cite{Holenstein2008,Peres2017}, \cite{De2017}, \cite{holden18}, \cite{Nazarov:2017}, \cite{holden2018lower}, \cite{chase2019new}). The problem of trace reconstruction can be stated simply as follows: consider a sequence $X$ which is simultaneously passed through $t$ independent deletion channels to yield $t$ output subsequences (also called \textit{traces}) of $X$ (see Fig.~\ref{fig:tdeletion}). How many such traces are needed to reconstruct $X$ perfectly? A variety of upper and lower bounds for this problem have been proposed, both for worst case and average case reconstruction. Our problem formulation is complementary to this, as we discuss next.
 
\begin{figure}[!h]
\begin{center}
\includegraphics[scale=0.2]{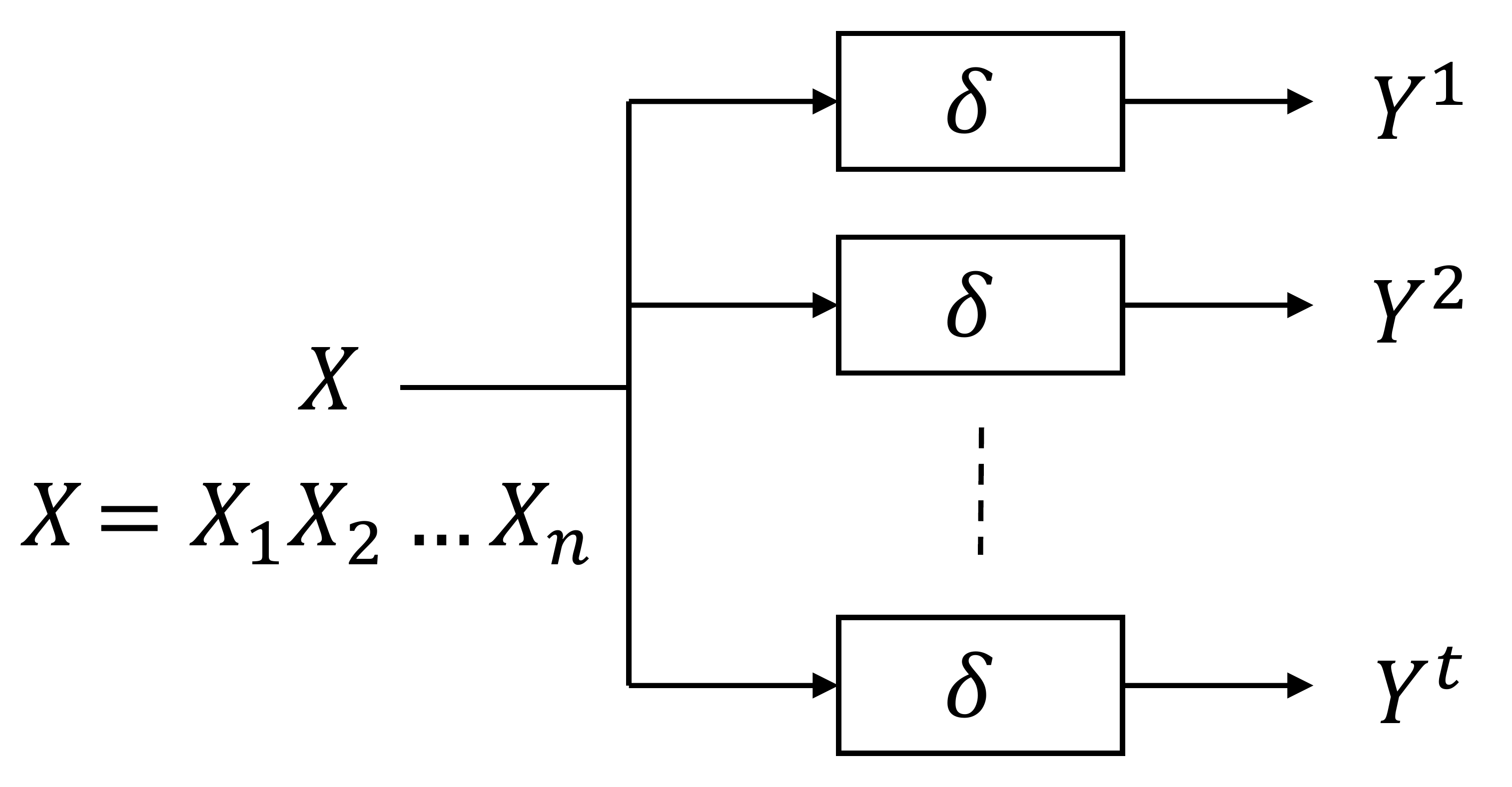}
\caption{The $t$-trace deletion channel model: the  sequence $X$ is passed through $t$ independent deletion channels to yield $t$ \textit{traces}. We aim to estimate $X$ from the $Y^{i}$s.}
\label{fig:tdeletion}
\end{center}
\end{figure}
 
\noindent \textbf{Problem formulation.}  Given an input sequence of length $n$ (known apriori), the independently and identically distributed (i.i.d.) deletion channel deletes each input symbol indepedently with probability $\delta$, producing at its output a subsequence of the input sequence. Consider a sequence $X$ passed through $t$ ($t$ is fixed) such deletion
channels as shown in
Fig.~\ref{fig:tdeletion}. We call this the $t$-trace deletion channel model. We ask four main questions:

\begin{figure}[!h]
\begin{center}
\includegraphics[scale=0.2]{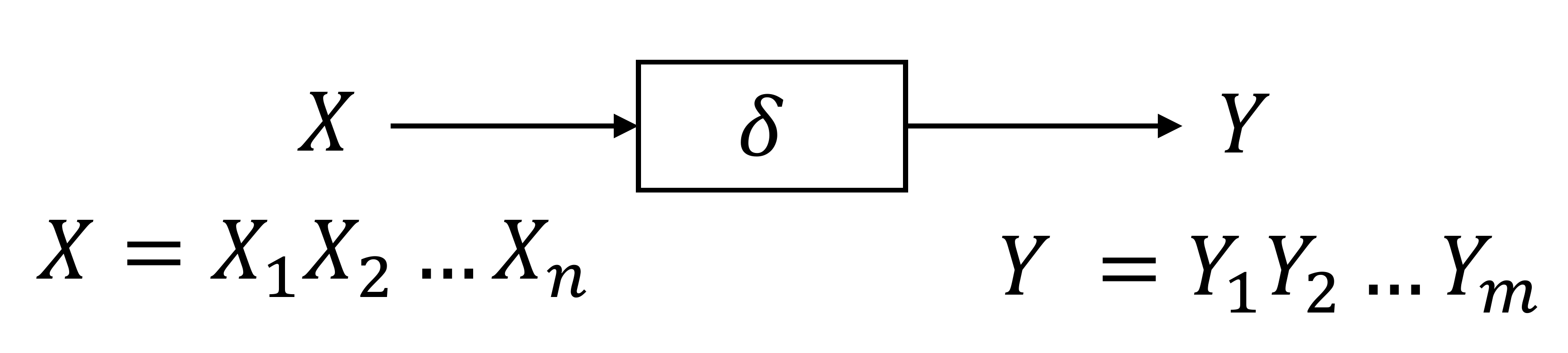}
\caption{The single-trace deletion channel model.}
\label{fig:1deletion}
\end{center}
\end{figure}
\begin{enumerate}[leftmargin = *]
\item \textbf{Sequencewise maximum-likelihood with one trace:} For $t=1$ (also called \textit{single-trace deletion channel}, see Fig.~\ref{fig:1deletion}), what is the maximum-likelihood estimate of $X$ having observed $Y=y$, i.e., a solution to $\argmax\limits_{x\in \{0,1\}^n}\ \Pr(Y=y|X=x)$.
\vspace{8pt}
\item \textbf{Sequencewise maximum-likelihood with multiple traces:}  For a fixed $t$, with $t>1$, what is the maximum-likelihood estimate of $X$ having observed $Y^{1}=y^1,Y^{2}=y^2,...,Y^{t}=y^t $, i.e., $$\argmax\limits_{x\in \{0,1\}^n}\ \Pr(Y^{1}=y^1,Y^{2}=y^2,...,Y^{t}=y^t|X=x).$$
\vspace{-12pt}
\item  \textbf{Symbolwise MAP with one trace:} For $t=1$  and  $X_i \sim\ \text{ind. Ber}(p_i)$ in Fig.~\ref{fig:1deletion}, what are the posterior distributions of $X_i$ given the trace $Y=y$, i.e., compute $\Pr(X_i=\alpha|Y=y)$.
\item \textbf{Symbolwise MAP with multiple traces:} For a fixed $t$, with $t>1$   and $X_i \sim\ \text{i.i.d. Ber}(0.5)$  in Fig.~\ref{fig:tdeletion}, what are the posterior distributions of $X_i$ given all traces $Y^{1}=y^1, Y^{2}=y^2,...,Y^{t}=y^t$, i.e., compute $\Pr(X_i=\alpha|Y^{1}=y^1, Y^{2}=y^2,...,Y^{t}=y^t)$.
\end{enumerate}

\vspace{5mm}

\noindent We make a few notes.

\begin{itemize}[leftmargin = *]
\item For a channel with memory such as the deletion channel, the symbolwise MAP/ML estimate and sequencewise MAP/ML estimate are not equivalent. For example, consider $t=1$, $n =6$ in Fig.~\ref{fig:1deletion} and say we observe the trace $Y = 1010$. The symbolwise MAP estimate with uniform priors for this case can be computed to be $\hat X_{smap} = 100110$ whereas the sequencewise ML estimate is $\hat X_{ml} = 101010$.
\item An answer to  3) above doesn't lead to a natural solution for 4) which is also due to deletion channels possessing memory. In particular, for a memoryless channel, we have $Y^{j}_i - X_i - Y^{k}_i$ and hence $\Pr(X_i=\alpha|Y^{j}, Y^{k}) \propto \Pr(Y^{j}_i, Y^{k}_i|X_i=\alpha) = \Pr(Y^{j}_i|X_i=\alpha) \Pr(Y^{k}_i|X_i=\alpha)\propto \Pr(X_i=\alpha| Y^{j}) \Pr(X_i=\alpha | Y^{k})$; so one could first obtain the posterior probabilities from each independent observation and combine them after. However, this is not the case for deletion channels since the markov chain $Y^{j}_i - X_i - Y^{k}_i$ no longer holds. As a result, one first needs to  ``align'' all the observations in order to compute the likelihoods.
\item Solving 2) and 4) naturally leads to two different algorithms for average-case trace reconstruction -- one that selects the most likely sequence $X$ and the other that selects the most likely value for each symbol $X_i$. However, the problem formulations in 3) and 4) ask a question complementary to that of trace reconstruction: given a fixed (possibly a few) number of traces, what is our ``best'' guess of $X$? The two problems 2) and 4) have different quantification of the word ``best''. Unlike trace reconstruction, we are not concerned with perfect reconstruction (since perfect reconstruction may not be possible with just a few traces). We also note that error rate guarantees for our algorithms (not a part of this work) would naturally lead to upper bounds for trace reconstruction.
\item The challenges associated with solving 1) and 2) and solving 3) and 4) are very different. On the one hand, solving 1) and 2) amounts to discovering alternate, equivalent or approximate formulations for the seemingly difficult discrete optimization problems. On the other hand, the challenge with 3) and 4) involves the design of efficient algorithms that are capable of exactly computing/approximating the symbolwise posterior probabilities, for which ``closed form'' expressions can be derived.
\end{itemize}

\vspace{5mm}
\noindent \textbf{Contributions.}
Our main contributions are as follows. 
\begin{itemize}[leftmargin=*]
\item We introduce mathematical tools and constructs to visualize and analyze single-trace and $t$-trace deletion error events (see Section~\ref{sec:notation}).  
\item For the single-trace deletion channel, we establish an equivalence between finding the optimal ML decoder and a continuous optimization problem we introduce (see Section~\ref{sec:1deletion_ML}).  This equivalence allows for the use of existing techniques for continuous optimization to be employed for a seemingly difficult discrete optimization problem.  This continuous optimization problem also turns out to be a signomial optimization.  Furthermore we also provide a polynomial time trace reconstruction heuristic with multiple traces that exploits this formulation. 
\item In Section~\ref{sec:1deletion}, we prove the following:
\begin{theorem}
For the single-trace deletion channel model with priors $X_i \sim \text{ind. Ber}(p_i)$ and observed trace $Y=y$, the symbolwise posterior probabilities $\Pr(X_i=1|Y=y)\ \forall\ i$ can be computed in $O(n^2)$ time complexity.
\end{theorem}
 \item  In Section~\ref{sec:exactsmap}, we prove the following:
\begin{theorem}
For the $t$-trace deletion channel model with priors $X_i \sim \text{i.i.d. Ber}(0.5)$ and observed traces $Y^{1}=y^1,...,Y^{t}=y^t$, the symbolwise posterior probabilities $\Pr(X_i = 1|Y^{1}=y^1,...,Y^{t}=y^t)\ \forall\ i$ can be computed in $O(2^t n^{t+2})$ time complexity.\\
\end{theorem}
\end{itemize}

\noindent \textbf{Tools and techniques.} In terms of theoretical tools, the series of books by Lothaire (\cite{lothaire1997combinatorics,lothaire2002algebraic,lothaire2005applied}) extensively use algebraic tools  for problems in the combinatorics of sequences (or \textit{words}), and our work is inspired by such techniques. We borrow some notation and leverage a few of their results in our work. \\

\noindent \textbf{Biological motivation.} Trace reconstruction in itself was motivated, in part, by problems in DNA sequence reconstruction. One such problem was to infer the DNA sequence of a common ancestor from the samples of its descendants. Our problem definition, that considers a fixed value of $t$, would fit naturally in a scenario with a fixed number of descendants where perfect reconstruction may not be possible.
Our motivation for considering this problem also comes from a recent DNA sequencing technology called \textit{nanopore sequencing}. The $t$-trace deletion channel model is a simplistic model to approximately capture the process of a DNA sequence passed through a nanopore sequencer\footnote{As seen in
  \cite{Mao2017},\cite{MDK17} there are more complicated effects of
  the nanopore reader not captured in this simple 
  representation.}.  \\

\noindent \textbf{More related work.}
Our work falls under the general umbrella of  sequence reconstruction over deletion channels (also see Levenshtein's work \cite{levenshtein2001efficient}), where we offer, to the best of our knowledge, the first non-trivial results on maximum likelihood and maximum aposteriori estimates for the single and multiple deletion channel. As mentioned earlier, the complementary problem of trace reconstruction falls closest to this work.

The deletion channel by itself is known to be notoriously difficult to analyse. As stated earlier, the capacity of a single deletion channel is still unknown (\cite{diggavi2007capacity,diggavi2006information,diggavi2001transmission}); as are optimal coding schemes. Prior works have looked at the design of codes for deletion channels (\cite{ratzer2005marker,ratzer2000codes,thomas2017polar}); these works consider use of a codebook (we do not). Statistical estimation over deletion channels is a difficult problem to analyze due its highly combinatorial nature. To the best of our knowledge, as yet there are no efficient estimation algorithms over deletion channels with statistical guarantees.

Very recently, a variant of the trace reconstruction problem called \textit{coded trace reconstruction} has been proposed, motivated by portable DNA-based data storage systems using DNA nanopores (see \cite{abroshan2019coding}, \cite{cheraghchi2019coded}, \cite{brakensiek2019coded}) and we believe that the ideas in this work may prove useful in such a setting.

There are other works on sequence assembly (see for example, \cite{Li09fastand}, \cite{Shomorony2016}),
where multiple short reads (from different segments of a sequence) are used to reconstruct the bigger
sequence. This work differs from sequence assembly since we are interested in inferring the entire length sequence and not  just small segments of it (which are then ``stitched'' together in sequence assembly).\\

\noindent \textbf{Paper Organization.} Section~\ref{sec:notation} introduces our notation and visualization tools for the single and $t$-trace  channel error events; Section~\ref{sec:1deletion_ML} provides a result concerning questions 1) and 2) wherein we prove the equivalence of ML decoding in question 1) to solving a continuous optimization problem; Section~\ref{sec:1deletion} answers question 3) for the single-trace channel;
 Section~\ref{sec:exactsmap}) answers question 4) for the  $t$-deletion channel;  Section~\ref{sec:Numerics} gives numerical evaluations; and Section~\ref{sec:conclusions} concludes the paper.

%% file: Sections_revised/notation.tex
\section{Notation and Tools}
\label{sec:notation}

\noindent \textbf{Basic notation:} We borrow some notation  from \cite{lothaire1997combinatorics} which deals with non-commutative algebra; we restate them here for convenience. Calligraphic letters refer to sets, capitalized letters
correspond to random variables and bold letters are used for functions.  Let $\mathcal{A}$ be the set of all
symbols. Throughout this work, we will focus on the case where $\mathcal{A}=\{0,1\}$,
though  our methods extend to arbitrarily large sets of finite size. Define
$\mathcal{A}^n$ to be the set of all $n$-length sequences and $\mathcal{A}^*$ to be the set of all finite length sequences with
symbols in $\mathcal{A}$. For a sequence $f$, $|f|$ denotes the length
of $f$. 

For integers $i,j$, we define $[i:j] \triangleq \{i,i+1,...,j\}$ if $j\geq i$ and $[i:j] \triangleq \emptyset$ otherwise. We also define $[i] \triangleq [1:i]$.

For a vector or sequence $x=(x_1,x_2,...,x_{i-1},x_i,x_{i+1},...,x_n)$, define $$x^{(i\rightarrow s)}\triangleq (x_1,x_2,...,x_{i-1},s,x_{i+1},...,x_n),$$ where the $i^{th}$ coordinate of $x$ is replaced by symbol $s$. 	\\

\noindent \textbf{Binomial coefficient (section 6.3 in \cite{lothaire1997combinatorics}):} Given sequences $f$ and $g$ in $\mathcal{A}^*$, the number of
subsequence patterns of $f$ that are equal to $g$ is called the
\textit{binomial coefficient} of $g$ in $f$ and is denoted by $f\choose g$.  For example, ${'apple' \choose 'ape'} = 2$ since $'ape'$ can be obtained from two (overlapping) subsequences of $'apple'$. This quantity has also been referred to as the \textit{embedding number} by another line of work \cite{elzinga2008algorithms}. For two sequences of lengths $n$ and $m$, the binomial coefficient can be computed using a dynamic programming approach in $O(nm)$ (see \cite{elzinga2008algorithms} or Proposition 6.3.2 in \cite{lothaire1997combinatorics}).
When the alphabet $\mathcal{A}$ is of cardinality 1, ${f \choose g} = {|f| \choose |g|}$, the classical binomial coefficient with their respective lengths as the parameters. This definition hence
could be thought of as a generalization of the classical binomial
coefficients. We will denote by $e$ the sequence of length 0, and {define ${f
  \choose e} \triangleq 1\ \forall\ f\ \in\ \mathcal{A}^*$.} We also define the classical binomial coefficient ${a \choose b}\triangleq 0,$ whenever $b>a$ or $b<0$ for ease of use.

The binomial coefficient forms the backbone for the probabilistic analysis of deletion channels since the input-output relation for a deletion channel (with deletion probability $\delta$, input $X$ and output $Y$) can be expressed as 
\begin{equation}
\label{eq:in-out_relation}
\Pr(Y=y|X=x) = {x \choose y} \delta^{|x|-|y|} (1-\delta)^{|y|}.
\end{equation}
The proof is straightforward -- the number of distinct error events that give rise to $y$ from $x$ is exactly the number of  subsequences of $x$ which are equal to $y$. Each of these error events has a probability $\delta^{|x|-|y|} (1-\delta)^{|y|}$, wherein the exponent of $\delta$ corresponds to the deleted symbols and the exponent of  $1-\delta$ to the undeleted symbols. \\

\noindent \textbf{Maximum Likelihood (ML) estimate:}
Given the definition of the binomial coefficient, the maximum-likelihood (ML) estimate over a deletion channel with observed output $Y=y$ can be cast in the following form:
\begin{align*}
\argmax_{x \in \{0,1\}^n}  {x \choose y}.\numberthis
\label{eq:ML_deletion}
\end{align*}
In the case of multiple deletion channels with observed traces $Y^{1}=y^1,...,Y^{t}=y^t$, the ML formulation is similar:
\begin{align*}
\argmax_{x \in \{0,1\}^n}  \prod_{j=1}^{t} {x \choose y^{j}}.\numberthis
\label{eq:ML_multiple_deletion}
\end{align*}
As yet, there is no known efficient way to come up with a solution for either of the above two formulations (see \cite{mitzenmacher2009survey}).\\

\noindent \textbf{Relaxed binomial coefficient.}
We now introduce the function  $\mathbf F(\cdot)$  which can be thought of as a real-valued relaxation of the binomial coefficient. This function is used in sections~\ref{sec:1deletion_ML} and~\ref{sec:1deletion}.

An intuitive definition is as follows: Consider a random vector $Z\in\{0,1\}^n$ such that $Z_i\sim$ ind. Ber$(p_i)$, and let $p$ be the vector of probabilities of length $n$. Then  $\mathbf F(p,v)=\mathbb E_{Z\sim p}\ {Z \choose v}$, i.e., $\mathbf F(p,v)$ is the expected number of times $v$ appears as a subsequence of $Z$. If $p \in \{0,1\}^n$, then $Z=p$ with probability 1 and $\mathbf F(p,v) = {p \choose v}$.
More precisely, $\mathbf F(\cdot)$ is defined as:
\begin{definition}
\label{def:f}
\begin{align*}
&\mathbf F: [0,1]^n \times \{0,1\}^m \rightarrow \mathbb{R},\\
\mathbf F(p,  v)\triangleq &\begin{cases} 
      \sum\limits_{\substack{\mathcal S|\mathcal S\subseteq [n],\\|\mathcal S|=m}} \quad \prod\limits_{i=1}^m p_{\mathcal S_i}^{v_i} (1-p_{\mathcal S_i})^{1-v_i} & 1 \leq m\leq n \\
	  1 &  0=m\leq n \\
      0 & \text{else}.  
   \end{cases}
\end{align*}
\end{definition}
Though at first sight $\mathbf F(p,v)$ sums over an exponential number of subsets, a dynamic programming approach can be used to compute it in $O(nm)$ time complexity (see Appendix~\ref{app:F_compute}). Note that this is the same complexity as computing the binomial coefficient.\\

\noindent \textbf{Decomposition of the $t$-trace deletion channel:}
 The following definitions and ideas are relevant to the results pertaining to multiple traces. We first state a result that aids in thinking about error events in multiple deletion channels.

The events occurring in the $t$-deletion channel model can be categorized into two groups:
\begin{enumerate}
\item an input symbol is deleted in \textit{all} the $t$-traces,
\item an input symbol is reflected in at least one of the traces.
\end{enumerate}
The error events of the first kind are in some sense ``not correctable'' or even ``detectable'' in any situation since it is impossible to tell with absolute certainty what and where the deleted symbol could have been (although the probabilities need not be uniform). The events of the second kind, however, can be detected and corrected in some situations. This thought process gives rise to a natural decomposition of the $t$-deletion channel model into a cascade of two channels: the first one being a deletion channel which captures error events of the first kind and the second one is what we call the \textit{remnant channel} which captures events of the second kind (see Fig.~\ref{fig:channel_equiv}). More precisely, we define the remnant channel as follows:
\begin{definition}
\textit{Remnant channel:} an input symbol to the remnant channel is reflected in any $k>0$ uniformly random traces and deleted in the rest with a probability ${t \choose k} \frac{\delta^{t-k}(1-\delta)^k}{1-\delta^t}$. Thus, the probability of an input symbol reflected in a \textit{fixed} set of $k>0$ traces is equal to $\frac{\delta^{t-k}(1-\delta)^k}{1-\delta^t}$.
\end{definition}
Note that probability of the union of all possible events here is $\sum_{k=1}^t {t \choose k} \frac{\delta^{t-k}(1-\delta)^k}{1-\delta^t}=1$, validating our definition.

\vspace{-5mm}

\begin{figure}[!h]
\begin{center}
\includegraphics[scale=0.4]{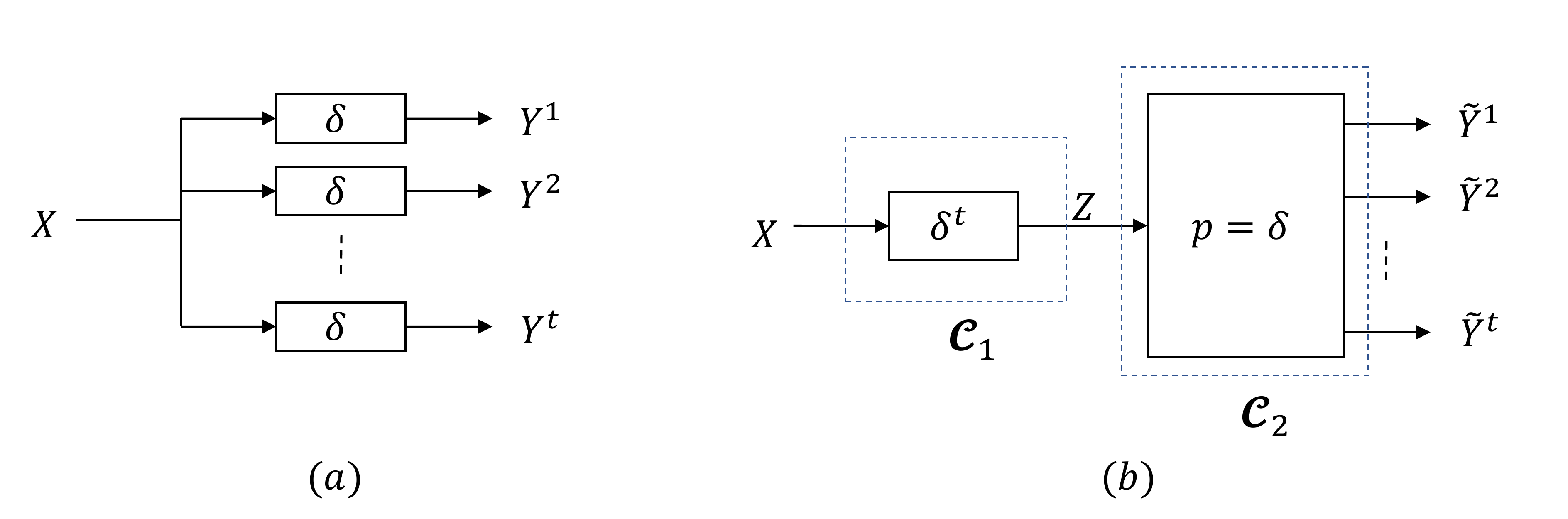}
\caption{A channel equivalence result: the $t$-trace deletion channel model in (a) is probabilistically equivalent to the the cascade of a deletion channel with the \textit{remnant channel} ($ \mathcal C_2$) in (b).}
\label{fig:channel_equiv}
\end{center}
\end{figure}
\vspace{-10mm}
\begin{restatable}{theorem}{channelequiv}
\label{thm:channel_equiv}
The $t$-deletion channel model and the cascade of the deletion channel with remnant channel shown in Fig.~\ref{fig:channel_equiv} are probabilistically equivalent, i.e., 
$$\Pr({Y}^{1}=y^1,{Y}^{2}=y^2,...,{Y}^{t}=y^t|X = x) = \Pr(\tilde{Y}^{1}=y^1,\tilde{Y}^{2}=y^2,...,\tilde{Y}^{t}=y^t|X = x).$$
\end{restatable}
A rigorous proof of this theorem for arbitrary length sequences  can be found in Appendix~\ref{app:channel_equiv}. A similar, though not equivalent, decomposition has been exploited in \cite{haeupler2014repeated} albeit for the purpose of characterizing the capacity of multiple deletion channels -- there the authors consider deletion patterns which are ``undetectable''; for example, a deletion in the deletion channel $\mathcal C_1$ in the cascade model is undetectable since none of the traces will reflect that input symbol. However, our channel decomposition result does not appear in \cite{haeupler2014repeated}.\\

\noindent\textbf{Edit graph} (\cite{Gusfield1997}): Similar graph constructs have been defined in related problems on common supersequences and subsequences (see \cite{Nicosia2001} for example). This graph is closely related to  the error events in the remnant channel. We start with a simple case and generalize subsequently. Define a directed graph called \textit{edit graph} given two sequences $f$ and $g$, where every path connecting the ``origin'' to the ``destination'' on the edit graph yields a supersequence $h$ of $f,g$, where $h$ is ``covered'' by $f,g$ -- i.e.,  each symbol of $h$ comes from either $f$ or $g$ or
both. In other words, given that $f$ and $g$ are the outputs of the remnant channel (with two outputs), each path from the origin of the edit graph to the destination corresponds to a possible input $h$ to the remnant channel and to an error event  which resulted in outputs $f,g$ with  input $h$.

For $f$ and $g$ in $\mathcal
A^*$, we form a directed graph $\mathcal{G}(f,g)$ with $(|f|+1)(|g|+1)$
vertices each labelled with a distinct pair $(i,j), 0\leq i\leq
|f|,\ 0\leq j\leq |g|$. A directed edge
$(i_1,j_1)\rightarrow(i_2,j_2)$ exists iff at least one of the following
holds: 
\begin{enumerate}
\item$i_2-i_1=1$ and $j_1=j_2$, or
\item $j_2-j_1=1$ and $i_1=i_2$, or
\item  $i_2-i_1=1$, $j_2-j_1=1$ and $f_{i_2}=g_{j_2}$, 
\end{enumerate}
where $f_i$ is the $i^{th}$ symbol of the sequence $f$. The origin is the vertex $(0,0)$ and the destination $(|f|,|g|)$.

\begin{figure}[!h]
\centering
\includegraphics[scale=0.25]{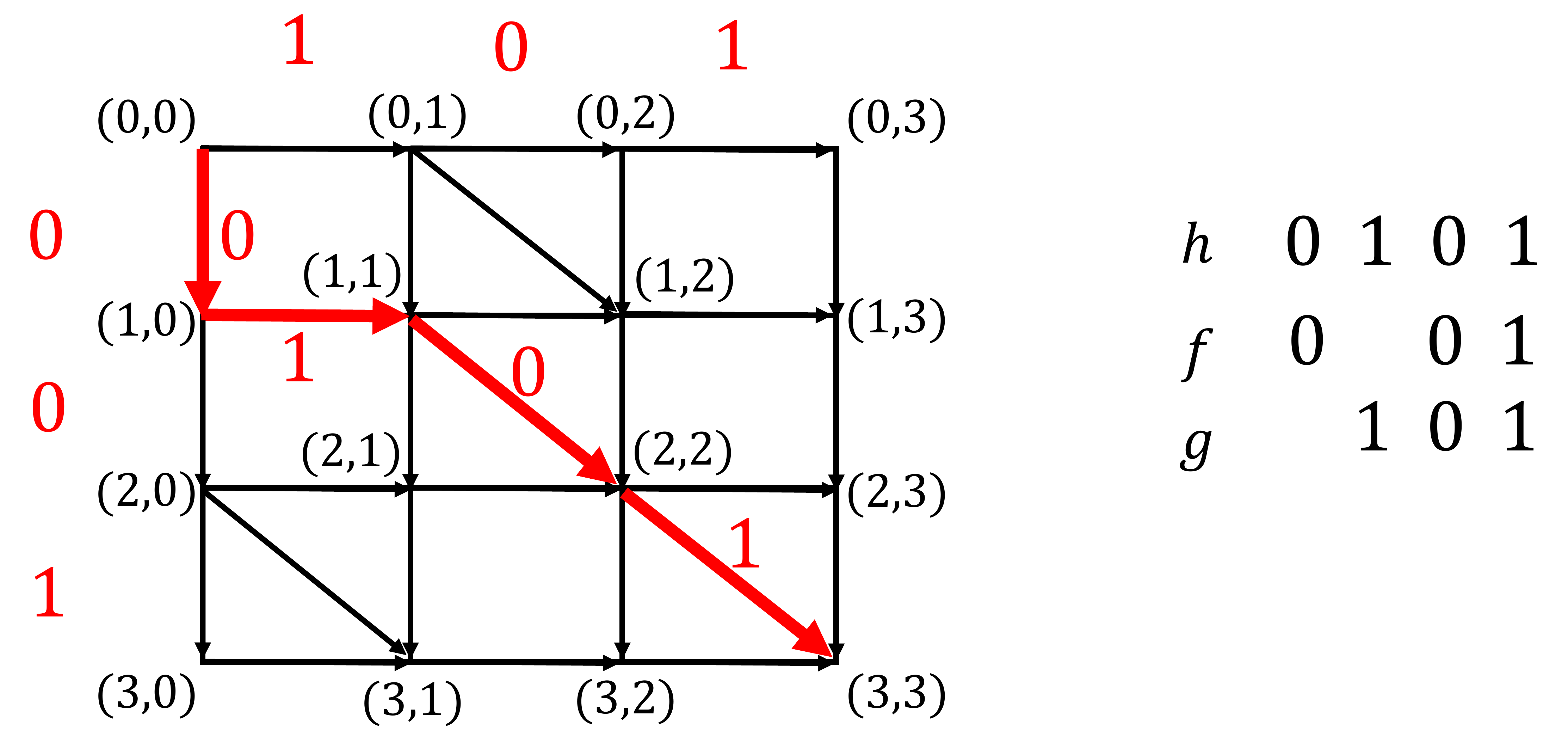}
\caption{ Edit graph for sequences $f=$ `001' and $g=$ `101'. Make a grid so the vertical edges are aligned with a symbol in $f$ and horizontal edges with $g$ as shown. A diagonal edge $(i{-}1,j{-}1)\rightarrow (i,j)$ exists if $f_i = g_j$. The thick red edges form a path from the origin to the destination; this path corresponds to $h=$`0101' -- sequentially append the corresponding symbol to which each edge is aligned. It can also be verified that $h$ is a supersequence of both $f$ and $g$, and could be obtained as a covering of $f$ and $g$; the path itself gives one such covering. This covering also corresponds to an error event (or a deletion pattern) in the remnant channel which would result in outputs $f$ and $g$ with input $h=$ `0101' -- the deletion pattern is shown in the figure.}
\label{fig:editgraph_smap1}
\end{figure}

Let $p=((i_1,j_1),(i_2,j_2),...,(i_m,j_m))$ be a path in $\mathcal{G}(f,g)$. We define $ s(p)$ to be the sequence corresponding to the path. Intuitively, $s(p)$ is formed by appending symbols in the following
way: append the corresponding $f$ symbol for a vertical edge, $g$ symbol for horizontal edge, and $f$ or $g$ symbol for  diagonal edge (see example Fig.~\ref{fig:editgraph_smap1}). Any path from $(0,0)$ to $(|f|,|g|)$ corresponds to a supersequence of $f$ and $g$ and which is covered by $f$ and $g$. More formally, define $ s(p)\triangleq x_1x_2...x_{m-1}$ where    
$$x_k = 
     \begin{cases}
       f_{i_{k+1}} \quad\text{if }j_{k}=j_{k+1},\\ 
       g_{j_{k+1}} \quad\text{if }i_{k}=i_{k+1},\\
       f_{i_{k+1}} \quad\text{else.}
    \end{cases}	
$$

The construct of edit graph can be extended to more than 2 sequences with the same idea. For sequences $f_1,f_2,...,f_t$, construct a $t$-dimensional grid with a number of vertices $(|f_1|+1)(|f_2|+1)...(|f_t|+1)$ labeled from $(0,0,...,0)$ to $(|f_1|,|f_2|,...,|f_t|)$. A vertex $u=(i_1,i_2,...,i_t)$ is connected to $v=(j_1,j_2,...,j_t)$ (we say $u \rightarrow v$) iff both of the following conditions are met:
\begin{itemize}
\item $j_l=i_l$ or $j_l=i_l+1$ $\forall\ l\in [t]$, i.e., $(i_1,...,i_t)$ and $(j_1,...,j_t)$ are vertices of a particular unit cube. Only these type of vertices can share an edge in the grid graph.
\item  Let $\mathcal T \subseteq [t]$ be the collection of indices where $j_l=i_l+1$. Then ${f_l}_{j_l}$ is equal $\forall\ l \in \mathcal T$. For example in 4 dimensional grid, consider the two vertices $(10,5,8,2)$ and $(10,6,9,2)$. In this case $\mathcal T = \{2,3\}$ since the second and third coordinates differ by 1. Therefore $(10,5,8,2)\rightarrow (10,6,9,2)$ iff ${f_2}_{5}={f_3}_{9}$. Note that if only one coordinate differs by 1 in the two vertices, a directed edge always exists (in other words all non-diagonal edges exist).
\end{itemize}
Define the vertex $(0,...,0)$ to be the origin of this graph and the vertex $(|f_1|,...,|f_t|)$ to be the destination. If $|f_j|=O(n)\ \forall\ j$, this graph has a number of vertices $O(n^t)$ and a maximum number of edges $O((2n)^t)$ since each vertex has at most $2^t-1$ outgoing edges.\\

\noindent\textbf{Infiltration product} (introduced in section 6.3 of \cite{lothaire1997combinatorics}):
The infiltration product has been extensively used in \cite{lothaire1997combinatorics}, as a tool in non-commutative algebra. Here, we give an edit-graph interpretation of this tool. A formal algebraic definition of the infiltration product is in Appendix~\ref{app:infil_def}. Using the edit graph we can construct the set of possible
supersequences $\mathcal{S}(f,g)$ of $f$, $g$ that are covered by the symbols in $f$ and $g$.
Indeed, multiple paths could yield the same supersequence and we
can count the number of distinct ways  $\mathbf N(h;f,g)$ one can construct
the same supersequence $h$ from $f$, $g$.
We can  informally define the
\emph{infiltration product $f\uparrow g$} of $f$ and $g$, as a polynomial with monomials the supersequences $h$ in
$\mathcal{S}(f,g)$ and coefficients  $\langle f\uparrow g,h\rangle$ equal to $\mathbf N(h;f,g)$. For the example in Fig.~\ref{fig:editgraph_smap1}, there is exactly one path corresponding to `101001' and hence $\langle 001\uparrow 101,101001 \rangle=1$ and similarly $\langle 001\uparrow 101,01001 \rangle=2$. One could find these coefficients for all relevant sequences and form the polynomial as described. We now give additional examples (see 6.3.14 in \cite{lothaire1997combinatorics}). Let $\mathcal{A}=\{a,b\}$, then
\begin{itemize}[wide=2pt]
\item $ab\uparrow ab=ab+2aab+2abb+4aabb+2abab$,
\item $ab\uparrow ba=aba+bab+abab+2abba+2baab+baba.$
\end{itemize}
The  infiltration operation is commutative and associative, and infiltration of two sequences $f\uparrow g$ is a polynomial with variables of length (or \textit{degree}) at most $|f|+|g|$; see \cite{lothaire1997combinatorics}.
The definition of infiltration extends to two polynomials via distributivity (precisely defined in Appendix~\ref{app:infil_def}), and consequently to multiple sequences as well. For multiple sequences, infiltration has the same edit graph interpretation: $\langle f_1\uparrow f_2 \uparrow...\uparrow f_t, w \rangle$ is the number of distinct ways of constructing $w$ as a supersequence of $f_1, f_2, ... ,f_t$ so that the construction covers $w$, i.e., construct the $t$-dimensional edit graph of $f_1, f_2, ... ,f_t$ and count the number of paths corresponding to $w$.
\\
{\small
\begin{center}
\begin{tabular}{ |P{3cm}|P{10cm}|  }
 \hline
 \multicolumn{2}{|c|}{Table of notation} \\
\hline
 $\mathcal A$ & A set \\
  \hline
 $X$    & A random variable or a random vector \\
  \hline
 $x$ &  A scalar or a vector variable\\
\hline
 $|x|$ & Length of the sequence $x$\\
 \hline
 $[i:j]$ & $\{i,i+1,...,j\}$\\
 \hline
 $x^{(i\rightarrow s)}$   & $(x_1,x_2,...,x_{i-1},s,x_{i+1},...,x_n)$   \\
  \hline
 ${f \choose g}$&   Binomial coefficient: number of subsequence patters of $f$ equal to $g$ \\
  \hline
 $\mathbf F(p,v)$ & Relaxed binomial coefficient: $\mathbb E_{Z\sim p} {Z \choose v}$ \\
  \hline
 $\langle f \uparrow g,h \rangle$    & Infiltration product: number of ways of obtaining sequence $h$ as a ``covered'' supersequence of $f$ and $g$ \\
  \hline
\end{tabular}
\end{center}} \vspace{2mm}

%% file: Sections_revised/single_ML.tex
\section{Sequencewise ML for the deletion channel}
\label{sec:1deletion_ML}
\subsection{A continuous optimization formulation for the single trace ML}
We here consider the single-trace  ML decoding in (\ref{eq:ML_deletion}),
assuming that the output sequence $Y=y$ is non-empty. 
To the best of our knowledge, the only known method to solve \eqref{eq:ML_deletion} involves solving a combinatorial optimization, essentially iterating over all possible choices of $x$ and computing the objective value for each of the choices.
The reason is that there seems to be no discernible pattern exhibited by the true ML sequence; as we see in the table below, the true ML sequence at times extends a few runs, and at times even introduces new runs! Here, we list a few examples of the trace and the corresponding 10-length ML sequences. 
\begin{center}
\begin{tabular}{ | c | c| } 
\hline
$y$ & The set of all $x_{ml}$ sequences \\ 
\hline
10111 & 1100111111 \\ 
\hline
1010 & 1101010100  \\ 
\hline
000100 & 0000001000, \quad 0000010000, \quad 0000011000 \\
\hline 
111101 & 1111111001,\quad 1111111011 \\
\hline
\end{tabular}
\end{center}

\vspace{5mm}

 In this section, we show that one could equivalently solve the continuous relaxation of \eqref{eq:ML_deletion} to obtain a solution for \eqref{eq:ML_deletion}. 
Before presenting the main result, we first state a useful lemma which factors a given coordinate $p_i$ out of the relaxed binomial coefficient $\mathbf F(p,y)$ we introduced in Definition~\ref{def:f}. 
\begin{restatable}{lemma}{deletionMLrellemma}
For $p=(p_1,p_2,..,p_i,...,p_n)$ and $Y=y=y_1...y_m$ with $n \geq m > 0$, we have
\begin{align*}
  \mathbf F(p,y) = \mathbf F( p_{
   [n]\backslash \{i\}},y) +  p_i \sum\limits_{k|y_k=1}\mathbf F( p_{[1:i-1]},  y_{[1:k-1]})\mathbf F( p_{[i+1:n]},  y_{[k+1,m]})\\ + (1-p_i) \sum\limits_{k|y_k=0}\mathbf F( p_{[1:i-1]},  y_{[1:k-1]})\mathbf F( p_{[i+1:n]},  y_{[k+1,m]}).
\end{align*}
\label{lemma:F_decomposition}
\end{restatable}
Recall that $\mathbf F(p,y)$ sums over all $m$-length subsets $\mathcal S$ and associates $p_{\mathcal S}$ with $y.$ Intuitively, this recursive relationship considers separately the cases where
\begin{itemize}
\item $i \notin \mathcal S$,
\item $i \in \mathcal S$ and is associated with a particular $y_k$ where $y_k = 1$,
\item $i \in \mathcal S$ and is associated with a particular $y_k$ where $y_k = 0$.
\end{itemize}
The detailed proof can be found in Appendix~\ref{app:F_lemma_proof}. It is clear from Lemma~\ref{lemma:F_decomposition} that $\mathbf F(p,y)$ is affine when projected onto each coordinate $p_i$. Thus, the extrema of $\mathbf F(p,y)$ must occur at the boundary of the support set of $p_i$; i.e., at either $p_i = 0$ or $p_i = 1$. Combining this with the fact that  $\mathbf F(\cdot)$ is a relaxed version of the binomial coefficient, we observe that the maximization problem in \eqref{eq:ML_deletion} is equivalent to its real-valued relaxation. The following result makes this precise.

\begin{theorem} The ML decoding problem for the single-trace deletion channel
\begin{equation}
\max_{x\in \{0,1\}^n} {x \choose y} 
\label{eq:ml_opt_equiv1}
\end{equation}
is equivalent to the problem
\begin{equation}
\max_{p\in [0,1]^n} \mathbf F(p,y).
\label{eq:ml_opt_equiv2}
\end{equation}
Furthermore, given any non-integral $p^* \in [0,1]^n$ that maximizes $\mathbf F(p,y)$, we can construct a corresponding integral solution $x^* \in \{0,1\}^n$ that maximizes $\mathbf F(x,y)$ and consequently also maximizes ${x \choose y}$.
\label{thm:ML_relaxation}
\end{theorem}

\begin{proof}
As noted earlier, we have ${x \choose y} = \mathbf F(x,y)$. Therefore, we are interested in proving the following:
\begin{align*}
\max_{x\in \{0,1\}^n} \mathbf F(x,y) \equiv \max_{p\in [0,1]^n} \mathbf F(p,y),\numberthis
\label{eq:ml_opt_equiv_proof1}
\end{align*}
where $\equiv$ refers to that the two problems are equivalent (have the same optimal objective value).
We prove this by applying the following claim.\\
\textbf{Claim:} Given any feasible $p=(p_1,p_2,...,p_i,...,p_n)$, at least one of the following holds true:
\begin{itemize}
\item$\mathbf F(p^{(i\rightarrow 0)},y) \geq \mathbf F(p,y)$. Recall from notation that $p^{(i\rightarrow 0)}=(p_1,p_2,...,p_{i-1},0,p_{i+1}...,p_n)$ is the vector where the $i^{th}$ coordinate is replaced by $0$.
\item $\mathbf F(p^{(i\rightarrow 1)},y) \geq \mathbf F(p,y)$.
\end{itemize}
Thus if $p^*$ is an optimal solution to \eqref{eq:ml_opt_equiv2} with $p_i\in (0,1)$, then at least one of $p^{(i\rightarrow 0)}$ or $p^{(i\rightarrow 1)}$ is also an optimal solution. Sequentially applying this argument for each coordinate of $p$ shows that there exists a point in $\{0,1\}^n$ which is an optimal solution to \eqref{eq:ml_opt_equiv2} and consequently to \eqref{eq:ml_opt_equiv1}.

It remains to prove our claim. We use Lemma~\ref{lemma:F_decomposition} to factor out $p_i$ terms in $\mathbf F(p,Y)$:
\begin{align*}
\mathbf F(p,y) = \mathbf F( p_{[n]\backslash \{i\}},y) +  p_i \sum\limits_{k|y_k=1}\mathbf F( p_{[1:i-1]},  y_{[1:k-1]})\mathbf F( p_{[i+1:n]},  y_{[k+1,m]})\\ + (1-p_i) \sum\limits_{k|y_k=0}\mathbf F( p_{[1:i-1]},  y_{[1:k-1]})\mathbf F( p_{[i+1:n]},  y_{[k+1,m]}).
\end{align*}
Now we express $\mathbf F(p^{(i\rightarrow 0)},y)$ and $\mathbf F(p^{(i\rightarrow 1)},y)$ as
$$\mathbf F(p^{(i\rightarrow 0)},y) = \mathbf F( p_{[n]\backslash \{i\}},y) +  \sum\limits_{k|y_k=0}\mathbf F( p_{[1:i-1]},  y_{[1:k-1]})\mathbf F( p_{[i+1:n]},  y_{[k+1,m]}),$$
$$\mathbf F(p^{(i\rightarrow 1)},y) = \mathbf F( p_{[n]\backslash \{i\}},y) +  \sum\limits_{k|y_k=1}\mathbf F( p_{[1:i-1]},  y_{[1:k-1]})\mathbf F( p_{[i+1:n]},  y_{[k+1,m]}).$$

\noindent Because $0\leq p_i\leq 1$ it directly follows that $$\min \left\{\mathbf F(p^{(i\rightarrow 0)},y),\mathbf F(p^{(i\rightarrow 1)},y)\right\} \leq \mathbf F(p,y) \leq  \max \left\{\mathbf F(p^{(i\rightarrow 0)},y),\mathbf F(p^{(i\rightarrow 1)},y)\right\},$$
thus proving our claim.

\end{proof}

The real-valued optimization problem in \eqref{eq:ml_opt_equiv2} falls under the umbrella of signomial optimization which is, in general, NP-hard (see for example, \cite{xu2014signomial}, \cite{chand2016signomial}). A standard technique for signomial optimization uses convexification strategies to approximate the optimal value. In particular, 
as stated in \cite{chand2016signomial}, the main observation underlying their methods is that certifying the nonnegativity of a signomial with at most \textit{one negative coefficient} can be accomplished efficiently. However, there are two problems with this approach in relation to our work -- 1. when expressed as a signomial optimization problem, \textit{all} the coefficients are negative in the ML optimization objective function,  and 2. the objective function has an exponential number of signomial terms as can be seen from Definition~\ref{def:f}. As a result, such strategies turn out to not be useful for the ML optimization problem. For instance, the techniques in \cite{chand2016signomial} resulted in the bound $\mathbf F(p,Y) \leq {|p| \choose |Y|}$ for most instances of $p$ and $Y$, where $|\cdot|$ denotes the length of the vector/sequence. This is a trivial bound that uses no information about $p$ and $Y$ other than their lengths. Moreover, with a slight change of variables, \eqref{eq:ml_opt_equiv2} could also be expressed as a maximization of a convex function in a convex set. With that being said, it is still unclear if \eqref{eq:ml_opt_equiv2} is solvable in polynomial time or not.


\subsection{ML via gradient ascent}
Given the continuous variable formulation of the ML problem in \eqref{eq:ml_opt_equiv2}, a natural heuristic to find an estimate of the ML sequence is to employ \textit{projected gradient ascent} to solve \eqref{eq:ml_opt_equiv2}. The algorithm, in short, can be described as follows (the exact algorithm is detailed as Alg.~\ref{alg:grad_asc}):\\
Step I: Start from a randomly chosen interior point (in our case, we start from $p=(0.5,0.5,...,0.5)$, the point corresponding to the uniform distribution).\\
Step II: Take a small step in the direction of the gradient $\nabla_p\ \mathbf F(p,y)$.\\
Step III: If the gradient step results in $p$ moving out of $[0,1]^n$, project it back onto $[0,1]^n$. Repeat Steps II and III until convergence.\\
Step IV: From the final $p$, determine the closest binary sequence to be the reconstructed sequence.

Moreover in Appendix~\ref{app:F_grad_comp}, we show using Lemma~\ref{lemma:F_decomposition} that $\nabla_p\ \mathbf F(p,y)$ can be computed in $O(n^2)$ as a ``by-product'' of computing $\mathbf F(p,y)$.

\begin{algorithm}[t!]
\caption{Single trace projected gradient ascent for ML}\label{alg:grad_asc}
\begin{algorithmic}[1]
\item Input: Blocklength $n$, Trace {$Y=y$}, Initial point $p = (p_1,p_2,...,p_n)$, step-size $\epsilon$, Max iterations $M$, Convergence criteria $C$ \\ Outputs: Estimated sequence $\hat X$
\State Iteration count $j=0$
\While {$C$ is FALSE and $j<M$} 
\State $p \leftarrow p + \epsilon \frac{\nabla_p \mathbf F(p,y)}{\mathbf F(p,y)}$
\State Replace $p_i \leftarrow 1$ for all $i:p_i > 1$
\State Replace $p_i \leftarrow 0$ for all $i:p_i < 0$
\State $j\leftarrow j+1$
\EndWhile
\State For each $i$, set $\hat X_i = \mathbbm 1 \{p_i>0.5\}$.
\State \textbf{return} $\hat X = \hat X_1 \hat X_2 ... \hat X_n$
\end{algorithmic}
\end{algorithm}

\newpage
\subsection{A heuristic for multiple traces}

The continuous variable ML formulation in \eqref{eq:ml_opt_equiv2} optimizes over the distributions $p$, instead of sequences $x$. In particular, we proved the following:
\begin{equation*}
\max_{x\in \{0,1\}^n} {x \choose y}  \equiv \max_{p\in [0,1]^n} \mathbf F(p,y) \equiv\ \max_{p\in [0,1]^n}\ \mathbb E_{Z \sim p} {Z \choose y}. 
\end{equation*}

At this point, one could ask how this formulation extends to multiple traces $Y^1=y^1,Y^2=y^2,...,Y^t=y^t$. The following theorem gives such a continuous optimization formulation with multiple traces.

\begin{theorem} The ML decoding with multiple traces 
\begin{equation}
\max_{x\in \{0,1\}^n} {x \choose y^1}{x \choose y^2}...{x \choose y^t} 
\label{eq:ml_opt_traces1}
\end{equation}
is equivalent to
\begin{equation}
\max_{p\in [0,1]^n} \mathbb E_{Z\sim p} \left[{Z \choose y^1}{Z \choose y^2}...{Z \choose y^t}  \right].
\label{eq:ml_opt_traces2}
\end{equation}
Furthermore, given any non-integral $p^* \in [0,1]^n$ that maximizes $\mathbb E_{Z\sim p} \left[{Z \choose y^1}{Z \choose y^2}...{Z \choose y^t}  \right]$, we can construct a corresponding integral solution $x^* \in \{0,1\}^n$ that also maximizes ${x \choose y^1}{x \choose y^2}...{x \choose y^t} $.
\label{thm:ML_multi_traces}
\end{theorem}

\begin{proof}
This theorem can be proved in the same way as Theorem~\ref{thm:ML_relaxation}, by showing that \\ $\mathbb E_{Z\sim p} \left[{Z \choose y^1}{Z \choose y^2}...{Z \choose y^t}  \right]$ is an affine function of each $p_i$; here we only prove this fact and the rest of the arguments follow exactly as in the proof of Theorem~\ref{thm:ML_relaxation}. 

To show this we use Lemma~\ref{lemma:bin_inf_relation} stated below; this Lemma is also closely related to the channel equivalence of Theorem~\ref{thm:channel_equiv} (see Appendix~\ref{app:bin_inf_lemma}).

\begin{restatable}{lemma}{bininfrelation}
\label{lemma:bin_inf_relation}
For $h,f_1,f_2,...,f_m \in \mathcal{A}^*$,\\
$${h \choose f_1} {h \choose f_2}...{h \choose f_m}=\sum_{w\in \mathcal{A}^*}\langle f_1\uparrow f_2\uparrow  ...\uparrow f_m,w \rangle{h \choose w}.$$
\end{restatable}

Using Lemma~\ref{lemma:bin_inf_relation}, we now have
\begin{align*}
\mathbb E_{Z\sim p} \left[{Z \choose y^1}{Z \choose y^2}...{Z \choose y^t}\right] &=  \mathbb E_{Z\sim p} \sum_{w\in \mathcal{A}^*}\langle y^1\uparrow y^2\uparrow  ...\uparrow y^t,w \rangle{Z \choose w}\\
&= \sum_{w\in \mathcal{A}^*}\langle y^1\uparrow y^2\uparrow  ...\uparrow y^t,w \rangle \mathbb E_{Z\sim p}  {Z \choose w}\\
& = \sum_{w\in \mathcal{A}^*}\langle y^1\uparrow y^2\uparrow  ...\uparrow y^t,w \rangle \mathbf F(p,w).
\end{align*}
Note that $\mathbf F(p,w)$ is affine in each $p_i$. Thus $\mathbb E_{Z\sim p} \left[{Z \choose y^1}{Z \choose y^2}...{Z \choose y^t}\right]$ is a linear combination of affine functions of each $p_i$, and hence is also affine in each $p_i$.
\end{proof}
The formulation of \eqref{eq:ml_opt_traces2}, by itself, is not very useful as it is unclear on how to efficiently compute $\mathbb E_{Z\sim p} \left[{Z \choose y^1}{Z \choose y^2}...{Z \choose y^t}\right]$. Indeed, if ${Z \choose y^i}\ \indep\ {Z \choose y^j}$, the expectation of products would decompose into the product $\prod_{j} \mathbb E_{Z\sim p} {Z \choose y^j} = \prod_{j} \mathbf F(p,y^j)$, and each of the terms in the product can be computed in $O(n^2)$ as detailed in Appendix~\ref{app:F_compute} -- this is however not the case as ${Z \choose y^i}$ and ${Z \choose y^j}$ are not independent. 

Having said that, we can now solve the maximization problem $\argmax_{p\in [0,1]^n} \prod_{j=1}^t \mathbf F(p,y^j)$ and hope that the resultant solution is also a good solution for  $\argmax_{p\in [0,1]^n} \mathbb E_{Z\sim p} \left[{Z \choose y^1}...{Z \choose y^t}\right]$; Algorithm~\ref{alg:grad_asc_traces} makes this idea precise. Moreover, instead of maximizing $\prod_{j=1}^t \mathbf F(p,y^j)$, we can further simplify the gradient computations by taking the log of the objective function, i.e., we solve $\argmax_{p\in [0,1]^n} \sum_{j=1}^t \log \mathbf F(p,y^j)$. This heuristic turns out to perform well in a variety of situations, as illustrated in Section~\ref{sec:Numerics}. As for the complexity, note that Alg.~\ref{alg:grad_asc_traces} involves the computation of $t$ gradients (each of which takes $O(n^2)$) at each gradient iteration. For a fixed number of max iterations $M$, the complexity of the algorithm is $O(n^2t)$.

\begin{algorithm}[t!]
\caption{Trace reconstruction heuristic via projected gradient ascent}\label{alg:grad_asc_traces}
\begin{algorithmic}[1]
\item Input: Blocklength $n$, Traces {$Y^1=y^1,Y^2=y^2,...,Y^t=y^t$}, Initial point $p = (p_1,p_2,...,p_n)$, step-size $\epsilon$, Max iterations $M$, Convergence criteria $C$ \\ Outputs: Estimated sequence $\hat X$
\State Iteration count $j=0$
\While {$C$ is FALSE and $j<M$} 
\State $p \leftarrow p + \epsilon \sum_{j=1}^t \frac{\nabla_p \mathbf F(p,y^j)}{\mathbf F(p,y^j)}$
\State Replace $p_i \leftarrow 1$ for all $i:p_i > 1$
\State Replace $p_i \leftarrow 0$ for all $i:p_i < 0$
\State $j\leftarrow j+1$
\EndWhile
\State For each $i$, set $\hat X_i = \mathbbm 1 \{p_i>0.5\}$.
\State \textbf{return} $\hat X = \hat X_1 \hat X_2 ... \hat X_n$
\end{algorithmic}
\end{algorithm}

%% file: Sections_revised/single.tex
\section{Symbolwise MAP for the single-trace deletion channel}
\label{sec:1deletion}

We here develop an algorithm to compute the symbolwise posterior probabilities for the single-trace deletion channel when the input symbols are independently generated with arbitrary priors. Consider the single deletion channel model in Fig.~\ref{fig:1deletion}, where $X=X_1...X_n$, each input symbol is generated $X_i \sim \text{ind. Ber}\ (p_i)$, and we observe  the trace $Y = y = y_1y_2...y_m$ with $m\leq n$. Define the vector of priors as  $p\triangleq(p_1,p_2,...,p_n)$. We first give an $O(n^2)$ algorithm to calculate the posterior probabilities $\Pr(X_i=1|Y=y)$, which in turn provides the symbolwise MAP estimate for the considered model. We then show how this algorithm can be used for trace reconstruction. We take three steps to present the algorithm.

\noindent \textbf{An expression for $\Pr(X_i=1|Y=y)$.} 
Let $\Pr(X_i=1)=p_i$. As a first step, we have
\vspace{6pt}
\begin{align*}
\Pr(X_i=1|{Y=y}) &= \frac{\Pr(X_i=1,Y=y)}{\Pr(Y=y)}
 = \frac{ \sum\limits_{\substack{ x| x_i=1}}  \Pr({X=x}) \Pr(Y=y|X=x)}{ \sum_{\substack{x}}  \Pr({X=x}) \Pr(Y=y|X=x)} \\
 &\overset{(a)}{=} \frac{ \sum\limits_{\substack{ x| x_i=1}}  \Pr({X=x}) { x \choose  y}}{ \sum_{\substack{x}}  \Pr({X=x}) { x \choose  y}}, \numberthis
\label{eq:approx_smap_1}
\end{align*}
where $(a)$ is because for a deletion channel $\Pr(Y=y|X=x)={x \choose y} \delta^{|x|-|y|}(1-\delta)^{|y|}$.
To proceed, we need to evaluate the summation in the numerator and the denominator.
Theorem~\ref{thm:approx_smap_postprob} expresses \eqref{eq:approx_smap_1} in terms of relaxed binomial coefficient terms $\mathbf F(\cdot)$. Recall that $\mathbf F(p,y) \triangleq \mathbb E_{X\sim p} {X \choose y}$, which is the denominator term in \eqref{eq:approx_smap_1}.

\begin{theorem}
\label{thm:approx_smap_postprob}
Let $X=X_1...X_n$ where $X_i \sim \text{ind. Ber}\ (p_i)$, and let $Y=y$ be the observed trace when $X$ is passed through a deletion channel. Then,
\begin{align*}
\Pr(X_i&=1|Y=y) = \frac{p_i}{\mathbf F( p,  y)} \left( \mathbf F( p_{[n]\backslash \{i\}},  y) +  \sum\limits_{k|y_k=1}\mathbf F( p_{[1:i-1]},  y_{[1:k-1]})\mathbf F( p_{[i+1:n]},  y_{[k+1,m]}) \right). \numberthis
\label{eq:smap_postprob_thm}
\end{align*}
\end{theorem}

\begin{proof}
The proof of this theorem employs the same trick used in the proof of Lemma~\ref{lemma:F_decomposition}. From \eqref{eq:approx_smap_1}, we have
\begin{align*}
\Pr(X_i = 1 | Y=y) = \frac{ \sum\limits_{\substack{ x| x_i=1}} \Pr({X=x}) { x \choose  y}}{\mathbf F(p,y)}.
\end{align*}
Now,
\begin{align*}
\sum_{\substack{ x| x_i=1}} & \Pr({X=x}) { x \choose y}
=\sum_{\substack{ x|x_i=1}} \Pr({X=x}) \sum_{\substack{\mathcal S\subseteq [n]\\ |\mathcal S|=m}} \mathbbm{1}\{ x_{\mathcal S}= y\}\\
&=\sum_{\substack{\mathcal S\subseteq [n]\\ |\mathcal S|=m}} \sum_{\substack{ x|x_i=1\\x_{\mathcal S}=y}}  \Pr({X=x}).\numberthis \label{eq:smapiter1}
\end{align*}
We first separate the outer summation into two cases: (a)  $\mathcal S|i \notin \mathcal S$ and (b)  $\mathcal S|i\in \mathcal S$.
We can express the first case as
\begin{align*}
&\hspace{-1cm}\sum_{\substack{\mathcal S\subseteq [n] \\ |\mathcal S|=m,i\notin \mathcal S}}\sum_{\substack{ x|x_i=1\\x_{\mathcal S}=y}}  \Pr({X=x})
=\sum_{\substack{\mathcal S\subseteq [n]\backslash \{i\}\\ |\mathcal S|=m}} \sum_{\substack{ x|x_i=1\\x_{\mathcal S}=y}}  \Pr({X=x})\\
&=\sum_{\substack{S\subseteq [n]\backslash \{i\}\\ |\mathcal S|=m}} \sum_{\substack{ x|x_i=1\\  x_{\mathcal S}= y}} \Big(\Pr(X_i=1)\Pr( X_{\mathcal S}= y)    \Pr( X_{[n]\backslash \mathcal S\cup\{i\}}= x_{[n]\backslash \mathcal S\cup\{i\}}) \Big)\\  
&=\sum_{\substack{\mathcal S\subseteq [n]\backslash \{i\}\\ |\mathcal S|=m}} p_i \Pr( X_{\mathcal S}= y)  \left(\sum_{\substack{ x|x_i=1\\  x_{\mathcal S}= y}}   \Pr( X_{[n]\backslash \mathcal S\cup\{i\}}= x_{[n]\backslash \mathcal S\cup\{i\}})\right)\\ 
&=\sum_{\substack{\mathcal S\subseteq [n]\backslash \{i\}\\ |\mathcal S|=m}} p_i \Pr( X_{\mathcal S}= y) \left(\sum_{(x_j|j\in [n]\backslash \mathcal S\cup \{i\})}   \Pr( X_{[n]\backslash \mathcal S\cup\{i\}}= x_{[n]\backslash \mathcal S\cup\{i\}})\right)\\
&=p_i \sum_{\substack{\mathcal S\subseteq [n]\backslash \{i\}\\ |\mathcal S|=m}}  \Pr( X_{\mathcal S}= y) = p_i \mathbf F( p_{[n]\backslash \{i\}},  y).\numberthis \label{eq:lemma3proof1}
\end{align*}

For the second term, we express the set $\mathcal S$ as a union  $\mathcal S = \mathcal S' \cup \{i\} \cup \mathcal S''$ such that $\mathcal S' \subseteq [i-1]$ and $\mathcal S'' \subseteq [i+1:n]$ to get:
\begin{align*}
&\sum_{\substack{\mathcal S\subseteq [n]\\ |\mathcal S|=m,\\i\in \mathcal S}} \sum_{\substack{ x|x_i=1\\x_{\mathcal S}=y}}  \Pr({X=x})= \sum_{k=1}^m\sum\limits_{\substack{\mathcal S\subseteq [n],\\|\mathcal S|=m,\\ \mathcal S_k = i}} \sum_{\substack{ x|x_i=1\\x_{\mathcal S}=y}}  \Pr({X=x})\\
&=\sum_{k=1}^m\sum_{\substack{\mathcal S'\subseteq [i-1]\\ |\mathcal S'|=k-1}}\sum_{\substack{\mathcal S''\subseteq [i+1:n]\\ |\mathcal S''|=m-k}} \sum_{\substack{ x|x_i=1\\x_{\mathcal S}=y}}\mathbbm{1}_{\{y_k=1\}}   \Pr({X=x}) \\
&=\sum_{k:y_k=1}\sum_{\substack{\mathcal S'\subseteq [i-1]\\ |\mathcal S'|=k-1}}\sum_{\substack{\mathcal S''\subseteq [i+1:n]\\ |\mathcal S''|=m-k}} \sum_{\substack{ x|x_i=1\\  x_{\mathcal S'}= y_{[1:k-1]}\\ x_{\mathcal S''}= y_{[k+1:m]}}} \Bigg ( \Pr(X_i=1) \Pr( X_{\mathcal S'}= y_{[1:k-1]})  
	\Pr( X_{\mathcal S''}= y_{[k+1:m]}) \\&\hspace{7cm} \Pr( X_{[n]\backslash \mathcal S'\cup \mathcal S''\cup \{i\}}= x_{[n]\backslash \mathcal S'\cup \mathcal S'' \cup\{i\}})\Bigg )\\
	&=p_i\sum_{k:y_k=1}\Bigg ( \Big( \sum_{\substack{\mathcal S'\subseteq [i-1]\\ |\mathcal S'|=k-1}}\Pr( X_{\mathcal S'}= y_{[1:k-1]})\Big) \Big(\sum_{\substack{\mathcal S''\subseteq [i+1:n]\\ |\mathcal S''|=m-k}}\Pr( X_{\mathcal S''}= y_{[k+1:m]})\Big ) \\
 & \hspace{5cm} \Big( \sum_{\substack{ x|x_i=1\\  x_{\mathcal S'}= y_{[1:k-1]}\\ x_{\mathcal S''}= y_{[k+1:m]}}} \Pr( X_{[n]\backslash \mathcal S'\cup \mathcal S''\cup \{i\}}= x_{[n]\backslash \mathcal S'\cup \mathcal S'' \cup\{i\}})\Big) \Bigg )\\ 
	&=p_i\sum_{k|y_k=1}\Bigg(\Big( \sum_{\substack{\mathcal S'\subseteq [i-1]\\ |\mathcal S'|=k-1}}\Pr( X_{\mathcal S'}= y_{[1:k-1]})\Big)	\Big( \sum_{\substack{\mathcal S''\subseteq [i+1:n]\\ |\mathcal S''|=m-k}}\Pr( X_{\mathcal S''}=y_{[k+1:m]}) \Big)\Bigg)\\
	&=p_i \sum\limits_{k|y_k=1}\mathbf F( p_{[1:i-1]},  y_{[1:k-1]})\mathbf F( p_{[i+1:n]},  y_{[k+1,m]}). \numberthis \label{eq:lemma3proof2}
\end{align*}
Plugging in \eqref{eq:lemma3proof1} and \eqref{eq:lemma3proof2} in \eqref{eq:approx_smap_1} proves the theorem.
\end{proof}

Alg.~\ref{alg:apprx_smap_dp} summarizes the computation of $\Pr(X_i=1|Y=y)$. 

\begin{algorithm}[t!]
\caption{Symbolwise posterior probabilities with one trace}\label{alg:apprx_smap_dp}
\begin{algorithmic}[1]
\item Input: Trace {$Y=y$}, priors $p$\\ Outputs: Posteriors $\Pr(X_i=1|Y=y)\ \forall\ i$
\State Compute $\mathbf F(p_{[1:k]},y_{[1:j]})\ \forall\ k,j$ and $\mathbf F(p_{[k:n]},y_{[j:m]})\ \forall\ k,j$ via Alg.~\ref{alg:F_comp}
\For {$i=1:n$}
\State Use \eqref{eq:smap_postprob_thm} to compute $\Pr(X_i=1|Y=y)$
\EndFor
\end{algorithmic}
\end{algorithm}

\noindent\textbf{A trace reconstruction heuristic with $t$ traces.}
The posterior probability computation in Alg.~\ref{alg:apprx_smap_dp} naturally gives rise to a trace reconstruction heuristic that  updates the symbolwise statistics sequentially on the traces, where we use Alg.~\ref{alg:apprx_smap_dp} with one trace at a time to continually update $\Pr(X_i=1|Y=y)$. The overall heuristic  is described in Alg.~\ref{alg:apprx_smap}. Note that the algorithm first needs to compute  $\mathbf F(p_{[1:k]},y_{[1:j]})\ \forall\ k,j$ and $\mathbf F(p_{[k:n]},y_{[j:m]})\ \forall\ k,j$ which requires $O(n^2)$ operations, as described in Appendix~\ref{app:F_compute}. Given this, the algorithm iterates over the $n$ indices and computes the posteriors in $O(n)$ for each of the index. Thus, the complexity of the algorithm is $O(n^2)$; note that $m=O(n)$ since $y$ is a deleted version of the input.

\begin{algorithm}[t!]
\caption{Trace reconstruction via iterative single-trace posterior probabilities}\label{alg:apprx_smap}
\begin{algorithmic}[1]
\item Input: Traces {$Y^{1}=y^1,...,Y^{t}=y^t$}, input length $n$ \\ Outputs: Estimate of the input $\hat X$
\State Initialize priors $p^{old}=p^{new} \gets (0.5,0.5,...,0.5)$
\For {$l=1:t$}
\State Use Alg.~\ref{alg:apprx_smap_dp} with $p^{old}$ and $y^{l}$ to update $p^{new}$
\State $p^{old}\gets p^{new}$
\EndFor
\For {$i=1:n$}
\If {$p^{new}_i\geq 0.5$}
 $\ \hat X_i \gets 1$
\Else
 $\ \hat X_i \gets 0$
\EndIf
\EndFor
\State \textbf{return} $\hat X_1 \hat X_2 ... \hat X_n$
\end{algorithmic}
\end{algorithm}

%% file: Sections_revised/multiple.tex
\section{Symbolwise MAP for the $t$-trace deletion channel}
\label{sec:exactsmap}

In this section, we put to use the ideas and constructs introduced in section~\ref{sec:notation} to exactly compute the symbolwise posterior probabilities given $t$-traces, which in turn gives a symbolwise MAP estimate with uniform input  priors  (motivated by average case trace reconstruction). With this formulation the symbolwise MAP with uniform priors can be seen as a minimizer of the symbol error rate in the context of average case trace reconstruction. In Appendix~\ref{app:remnant_postprob}, we also provide a method to compute the symbolwise posterior probabilities for the remnant channel -- we encourage the reader  to use this appendix as a warm-up. For the $t$-trace deletion channel, similar expressions arise due to the channel equivalence result of Theorem~\ref{thm:channel_equiv}.

Let $\mathcal A = \{0,1\}$, and  assume that $X\sim $ Uniform $\mathcal A^n$. Our goal is to compute the symbolwise posterior probabilities $\Pr(X_i=1|Y^1=y^1,...,Y^t=y^t)$, where $Y^j$ is the $j^{th}$ trace. Our proposed algorithm is provided in Alg.~\ref{alg:exact_smap} and estimates the symbolwise MAP  (with uniform priors). We can  directly leverage   Alg.~\ref{alg:exact_smap} to reconstruct the input as follows: for each index $i$, compute $\Pr(X_i =1|Y^1=y^1,...,Y^t=y^t)$ and decide
\begin{align*}
\hat{X}_i = \begin{cases}
 1,\quad &$if$\  \Pr(X_i =1|Y^1=y^1,...,Y^t=y^t)\geq 0.5 \\
0, \quad &$otherwise$.
\end{cases}
\end{align*}
Through the rest of this section, we show how to compute $\Pr(X_i =1|Y^1=y^1,...,Y^t=y^t)$\footnote{Symbolwise MAP with non-uniform priors is  part of on-going work.} in two steps:
\begin{itemize}
\item We first give an expression for $\Pr(X_i =1|Y^1=y^1,...,Y^t=y^t)$ which sums over potentially an exponential number of terms.
\item We then show that this summation can be computed in polynomial time (polynomial in the blocklength $n$).
\end{itemize}

\noindent \textbf{Step 1: An expression for $\Pr(X_i =1|Y^1=y^1,...,Y^t=y^t)$.}
\begin{theorem}
\label{thm:exactSMAP_posteriorprob}
Assume $X\sim $ Uniform $\mathcal A^n$ or equivalently $X_i \sim\ \text{Ber}(0.5)$. The posterior probability of the $i^{th}$ bit given the $t$ traces can be expressed as
\begin{align*}
\Pr(X_i =1|Y^1=y^1,...,Y^t&=y^t) \\ =
 & \Bigg[ \sum_{k=0}^n 2^{n-k-1}{n-1 \choose k} \sum_{w||w|=k}  \langle y^1 \uparrow ... \uparrow y^t,w \rangle \\
&+ \sum_{k=0}^n \sum_{j=1}^k 2^{n-k}{i-1 \choose j-1}{n-i \choose k-j} \sum_{\substack{w| |w|=k,\\w_j=1}} \langle  y^1 \uparrow ... \uparrow y^t,w \rangle \Bigg] \Big/ \\ 
 &\Bigg[ \sum_{k=0}^n 2^{n-k} {n \choose k} \sum_{\substack{w| |w|=k}} \langle  y^1 \uparrow ... \uparrow y^t,w \rangle  \Bigg]. \numberthis
 \label{eq:posterior_prob}
\end{align*}
\end{theorem}
Note that the summation index, $w| |w|{=}k$ is over all sequences $w$ of length $k$; this is an alternate expression for $w|w{\in}\mathcal A^k$. We follow this convention throughout the rest of the paper. 
\begin{proof}
\begin{align*}
\Pr(X_i=1&|Y^1=y^1,...,Y^t=y^t) = \sum_{\substack{x||x|=n,\\x_i=1}} \Pr(X=x|Y^1=y^1,...,Y^t=y^t)\\
&\overset{(a)}{=} \frac{1}{2^n \Pr(Y^1=y^1,...,Y^t=y^t)} \sum_{\substack{x||x|=n,\\x_i=1}} \Pr(Y^1=y^1,...,Y^t=y^t|X=x)\\
&\overset{(b)}{=} \frac{1}{2^n \Pr(Y^1=y^1,...,Y^t=y^t)} \sum_{\substack{x||x|=n,\\x_i=1}} \prod_{j=1}^t\Pr(Y^{j}=y^j|X=x),
\end{align*}
where $(a)$ uses Bayes' principle and $(b)$ is because each deletion channel acts independently. Recall that for  a deletion channel with deletion probability $\delta$, $\Pr(Y=y|X=x)={x \choose y}\delta^{|x|-|y|}(1-\delta)^{|y|}$. Also, using the fact that $\Pr(Y^1=y^1,...,Y^t=y^t)=\sum\limits_{\substack{x||x|=n}}\Pr(x) \Pr(Y^1=y^1,...,Y^t=y^t|X=x)$ we have,
\begin{align*}
\Pr(X_i=1|Y^1=y^1,...,Y^t=y^t)=  \frac{\sum\limits_{\substack{x||x|=n,\\x_i=1}} {x \choose y^1}...{x \choose y^t}}{\sum\limits_{\substack{x||x|=n}} {x \choose y^1}...{x \choose y^t}}. \numberthis \label{eq:thmexactsmap_proofterm0}
\end{align*}

We first simplify the numerator $\sum\limits_{\substack{x||x|=n,\\x_i=1}} {x \choose y^1}...{x \choose y^t}$; the denominator can be simplified using the same approach. Now,

\begin{align*}
\sum\limits_{\substack{x||x|=n,\\x_i=1}} {x \choose y^1}...{x \choose y^t}  &\overset{(a)}{=} \sum_{\substack{x||x|=n,\\x_i=1}}  \sum_{w\in \{0,1\}^*} {x \choose w} \langle y^1  \uparrow ... \uparrow y^t,w \rangle \\
  &=\sum_{w\in \mathcal A^*} \langle y^1  \uparrow ... \uparrow y^t,w \rangle \sum_{\substack{x||x|=n,\\x_i=1}}{x \choose w}\\
&\overset{(b)}{=}\sum_{w\in \mathcal A^*} 2^{n-|w|}  \langle y^1  \uparrow ... \uparrow y^t,w \rangle   \left(\frac{1}{2}{n-1 \choose |w|}+\sum_{j|w_j=1}{i-1 \choose j-1}{n-i \choose |w|-j}\right)
\end{align*}
where $(a)$ is due to Lemma~\ref{lemma:bin_inf_relation} and $(b)$ due to Lemma~\ref{lemma:smapsum} (both introduced in \cite{Srini2018}); see Appendix~\ref{app:bin_inf_lemma} and Appendix~\ref{app:smapsum} for the statement and proof. 

\noindent Therefore we have,
\begin{align*}
\sum\limits_{\substack{x||x|=n,\\x_i=1}} {x \choose y^1}...{x \choose y^t} 
&\overset{(a)}{=} \sum_{k=0}^{\infty} 2^{n-k-1}{n-1 \choose k} \sum_{w| |w|=k}  \langle y^1  \uparrow ... \uparrow y^t,w \rangle \\
&\hspace{1cm}+ \sum_{k=0}^{\infty} \sum_{j=1}^k 2^{n-k}{i-1 \choose j-1}{n-i \choose k-j} \sum_{\substack{w| |w|=k,\\w_j=1}} \langle y^1  \uparrow ... \uparrow y^t,w \rangle \\ 
&\overset{(b)}{=} \sum_{k=0}^n 2^{n-k-1}{n-1 \choose k} \sum_{w| |w|=k}  \langle y^1  \uparrow ... \uparrow y^t,w \rangle \\
&\hspace{1cm}+ \sum_{k=0}^n \sum_{j=1}^k 2^{n-k}{i-1 \choose j-1}{n-i \choose k-j} \sum_{\substack{w| |w|=k,\\w_j=1}} \langle y^1  \uparrow ... \uparrow y^t,w \rangle, \numberthis
\label{eq:thmexactsmap_proofterm1}
\end{align*}
where in $(a)$ we first fix $|w|$ and then sum over all $w$ of the given length and $(b)$ holds because the combinatorial terms are $0$ when $k>n$. 
A similar analysis gives  \begin{align*}
\sum\limits_{x| |x|=n} &{x \choose y^1}...{x \choose y^t}  = \sum_{k=0}^n 2^{n-k} {n \choose k} \sum_{\substack{w| |w|=k}} \langle y^1  \uparrow ... \uparrow y^t,w \rangle.\numberthis
\label{eq:thmexactsmap_proofterm2}
\end{align*} 
Plugging \eqref{eq:thmexactsmap_proofterm1} and \eqref{eq:thmexactsmap_proofterm2} in \eqref{eq:thmexactsmap_proofterm0}, we get the expression in Theorem~\ref{thm:exactSMAP_posteriorprob},
\begin{align*}
\Pr(X_i =1|Y^1=y^1,...,&Y^t=y^t) \\ =
 & \Bigg[ \sum_{k=0}^n 2^{n-k-1}{n-1 \choose k} \sum_{w||w|=k}  \langle y^1  \uparrow ... \uparrow y^t,w \rangle \\
&+ \sum_{k=0}^n \sum_{j=1}^k 2^{n-k}{i-1 \choose j-1}{n-i \choose k-j} \sum_{\substack{w| |w|=k,\\w_j=1}} \langle y^1  \uparrow ... \uparrow y^t,w \rangle \Bigg] \Big/ \\ 
 &\Bigg[ \sum_{k=0}^n 2^{n-k} {n \choose k} \sum_{\substack{w| |w|=k}} \langle y^1  \uparrow ... \uparrow y^t,w \rangle \Bigg].
\end{align*}
\end{proof}

\noindent \textbf{Step 2: Dynamic program to compute $\sum\limits_{w| |w|=k}  \langle y^1  \uparrow ... \uparrow y^t,w \rangle$ and $\sum\limits_{\substack{w| |w|=k,\\w_j=1}} \langle y^1  \uparrow ... \uparrow y^t,w \rangle$.}
  Note that the number of sequences $w$ such that $|w|=k$ is $O(2^k)$ so a naive evaluation is exponential in the blocklength $n$. We can, however, exploit the edit graph to come up with a dynamic program resulting in an algorithm which is polynomial in $n$. 

Recall that in the edit graph, $\langle y^1  \uparrow ... \uparrow y^t,w \rangle$ is equal to the number of distinct paths from the origin $(0,...,0)$ to the destination $(|y^1|,...,|y^t|)$ and which correspond to $w$. Hence,
\begin{enumerate}[wide=0pt]
\item[(a)]  $\sum\limits_{w| |w|=k}  \langle y^1  \uparrow ... \uparrow y^t,w \rangle$ is the number of distinct paths of length $k$ from origin to destination and,
\item[(b)]  $\sum\limits_{\substack{w| |w|=k,\\w_j=1}} \langle y^1  \uparrow ... \uparrow y^t,w \rangle$ is the number of such paths of length $k$ such that the $j^{th}$ edge of the path corresponds to a `1'.
\end{enumerate}

\noindent With this interpretation, the dynamic program for (a) follows naturally -- the number of  $k$-length paths from the origin to any vertex is the sum of the number of $(k{-}1)$-length paths from the origin to all incoming neighbors of the vertex. To make this formal, associate a  polynomial (in $\lambda$) for each vertex, such that the coefficient of $\lambda^k$ 
is equal to the number of paths of length $k$ from the origin to $v$: we call it the "forward-potential" polynomial $p^{for}_v(\lambda)$ for vertex $v$, the coefficient of $\lambda^k$ as earlier is denoted by $\langle p^{for}_v(\lambda),\lambda^k \rangle $. The dynamic program to compute $p^{for}_v(\lambda)$ for all $v$ can be expressed as:
\begin{equation}
p^{for}_v(\lambda) = \sum_{u|u\rightarrow v} \lambda p^{for}_u(\lambda).
\end{equation}
With this definition, we have $$\sum\limits_{w| |w|=k}  \langle y^1  \uparrow ... \uparrow y^t,w \rangle =\langle p^{for}_{destination}(\lambda),\lambda^k \rangle.$$ In the example in Fig.~\ref{fig:editgraph_smap1}, one could do the following: order the vertices $(0,0)$ to $(3,3)$ lexicographically and then compute $p^{for}_v(\lambda)$ in the same order. Because of the directed grid nature of the edit graph, every vertex has incoming neighbors which are lexicographically ahead of itself. Also we initialize $p^{for}_{(0,0)}(\lambda)=1$.  For the example in Fig.~\ref{fig:editgraph_smap1}, the forward-potentials are shown in Fig.~\ref{fig:editgraph_smap2}. The complexity of this dynamic program is $O(2^tn^{t+1})$ as it goes over $O(n^t)$ vertices and for each vertex it sums $O(2^t)$ polynomials, each of degree $O(n)$.

\begin{figure}[!h]
\centering
\includegraphics[scale=0.25]{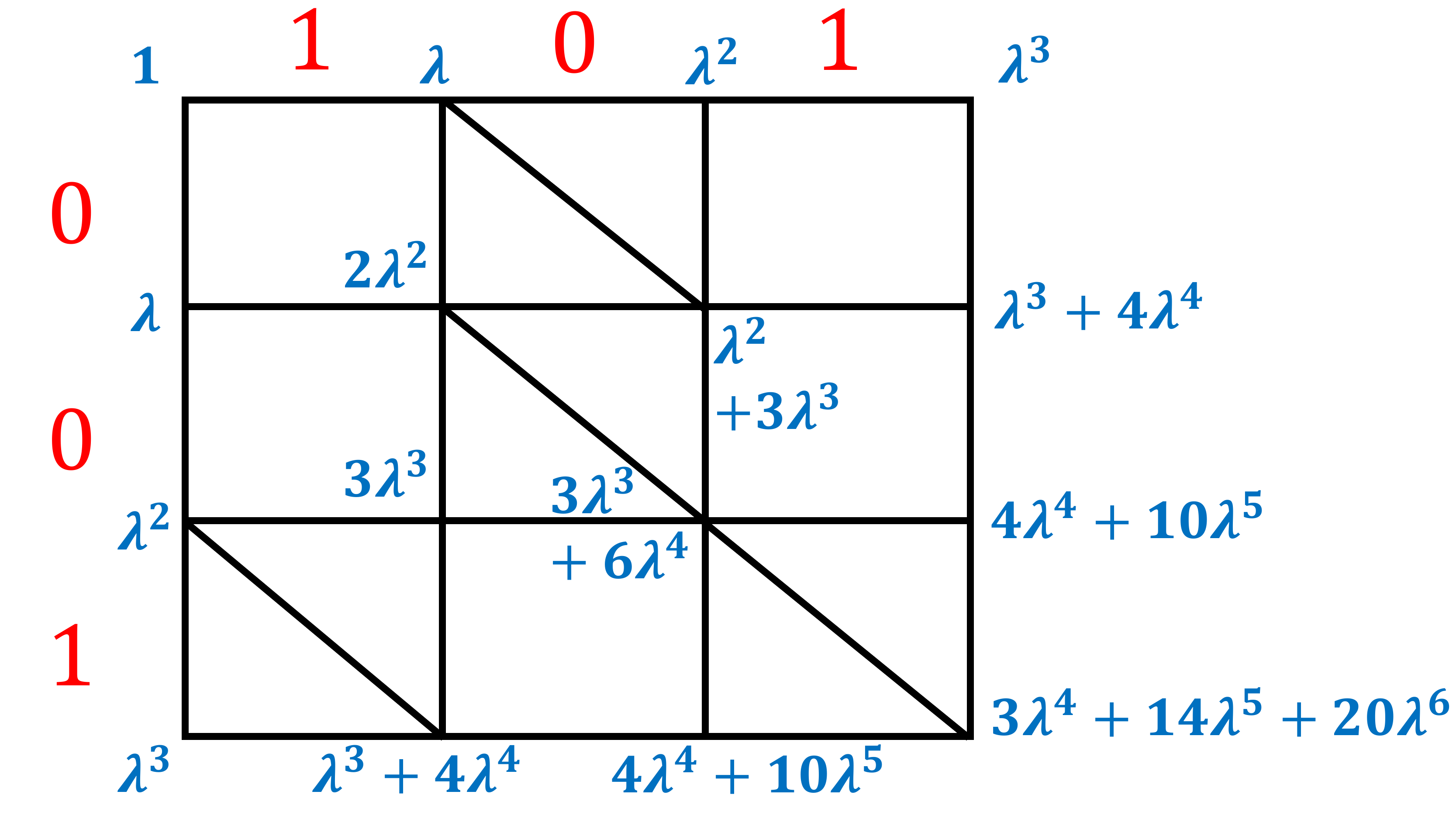}
\caption{The forward-potential $p^{for}_v(\lambda)$ at each vertex.} 
\label{fig:editgraph_smap2}
\end{figure}

We compute (b) as follows: pick an edge $(u{\rightarrow}v)$ which corresponds to `1', count the number of $(j{-}1)$-length paths from origin to $u$ and multiply it with the number of $(k{-}j)$-length paths from $v$ to the destination -- this is exactly the number of paths of length $k$ such that its $j^{th}$ edge is $(u{\rightarrow}v)$. Summing this term for all such edges which correspond to 1 gives us the term in (b). Note that we have already computed the number of $k$-length paths ($\forall k$) from origin to every vertex  in $p^{for}_v(\lambda)$ . We can similarly compute the number of $k$-length  paths ($\forall k$) from every vertex to the destination as $p^{rev}_v(\lambda)$ -- the "reverse potential" polynomial. The dynamic program for $p^{rev}_v(\lambda)$ is:
\begin{equation}
p^{rev}_v(\lambda) = \sum_{u|v\rightarrow u} \lambda p^{rev}_u(\lambda),
\end{equation}
with $p^{rev}_{destination}(\lambda)=1$.  The reverse potentials for the example in Fig.~\ref{fig:editgraph_smap1} is shown in Fig.~\ref{fig:editgraph_smap3}. Like in the case of forward potential, we first order the vertices reverse lexicographically and then invoke the dynamic program above sequentially to compute the reverse potential polynomial at each vertex.

\begin{algorithm}[t!]
\caption{Computing the forward-potentials $p^{for}_u(\lambda)$ }\label{alg:forward_pot}
\begin{algorithmic}[1]
\item Input: Edit graph {$\mathcal G(y^1,...,y^t)$}\\ Outputs: $p^{for}_v(\lambda)\ \forall\ v$
\State Order the vertices from $(0,0,...,0)$ to $(|y^1|,|y^2|,...,|y^t|)$ lexicogaphically; let the ordered list be $\mathcal V$
\State Initialise $p^{for}_{(0,...,0)}(\lambda)\gets 1$
\For{$v\ \in\ \mathcal V$}
\State \textbf{assign} $p^{for}_v(\lambda)\gets  \sum_{u|u\rightarrow v} \lambda p^{for}_u(\lambda)$
\EndFor
\end{algorithmic}
\end{algorithm}

\begin{algorithm}[t!]
\caption{Computing the reverse-potentials $p^{rev}_u(\lambda)$ }\label{alg:reverse_pot}
\begin{algorithmic}[1]
\item Input: Edit graph {$\mathcal G(y^1,...,y^t)$}\\ Outputs: $p^{rev}_v(\lambda)\ \forall\ v$
\State Order the vertices from $(|y^1|,|y^2|,...,|y^t|)$ to $(0,0,...,0)$ reverse lexicogaphically; let the ordered list be $\mathcal V$
\State Initialise $p^{rev}_{(|y^1|,|y^2|,...,|y^t|)}(\lambda)\gets 1$
\For{$v\ \in\ \mathcal V$}
\State \textbf{assign} $p^{rev}_v(\lambda) \gets \sum_{u|v\rightarrow u} \lambda p^{rev}_u(\lambda)$
\EndFor
\end{algorithmic}
\end{algorithm}

\begin{figure}[!h]
\centering
\includegraphics[scale=0.25]{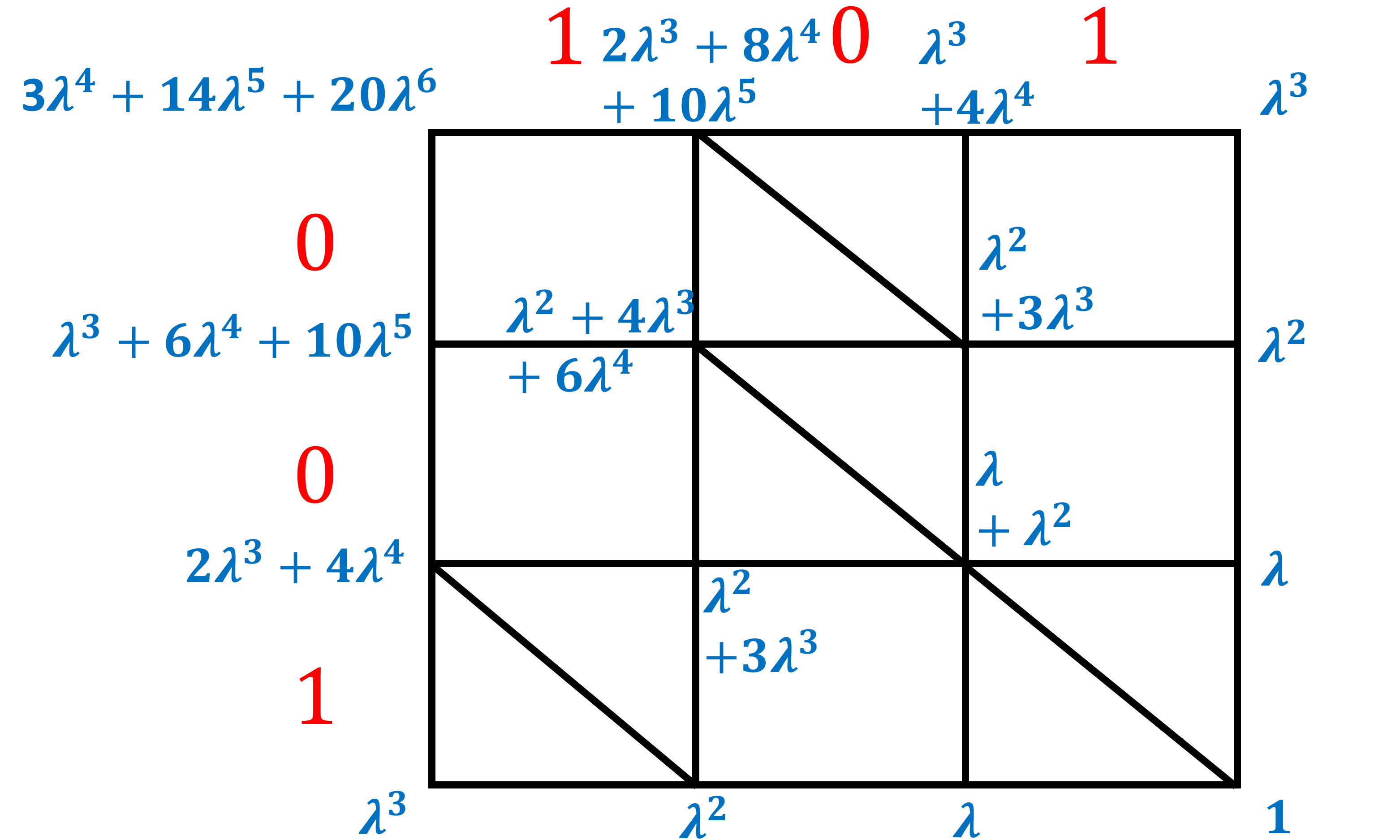}
\caption{The reverse-potential $p^{rev}_v(\lambda)$ at each vertex.} 
\label{fig:editgraph_smap3}
\end{figure}

With this, the term in (b) can be expressed as:

\begin{align*}
\sum\limits_{\substack{w| |w|=k,\\w_j=1}}\langle y^1 \uparrow ... \uparrow y^t,w \rangle =\hspace{-2mm} \sum_{\substack{(u,v)|\\s(u\rightarrow v)=1}} \hspace{-2mm}\langle p^{for}_u(\lambda),\lambda^{j-1} \rangle \langle p^{rev}_v(\lambda),\lambda^{k-j} \rangle.
\end{align*}

Alg.~\ref{alg:exact_smap} now summarizes the computation of the posterior probabilities. This algorithm iterates over all the edges (we have $O((2n)^t)$ of these), and also $k,j$ ($O(n)$ each). The time complexity of Alg.~\ref{alg:exact_smap} hence is $O(2^tn^{t+2})$.

\begin{algorithm}[t!]
\caption{Symbolwise MAP with $t$ traces} 
\label{alg:exact_smap}
\begin{algorithmic}[1]
\item Input: Traces {$Y^1=y^1,...,Y^t=y^t$, input length $n$}\\ Output: $\hat X = \hat X_1\hat X_2...\hat X_n$
\State Construct edit graph $\mathcal G(y^1,...,y^t)$
\State Use Alg.~\ref{alg:forward_pot} and Alg.~\ref{alg:reverse_pot} on $\mathcal G(y^1,...,y^t)$ to calculate $p^{for}_v(\lambda)$ and $p^{rev}_v(\lambda)$ $\forall\ v$

\For{$k \in\ [0:n]$}
\State \textbf{assign} $\sum\limits_{w| |w|=k}  \langle y^1 \uparrow ... \uparrow y^t,w \rangle \gets \langle p^{for}_{destination}(\lambda),\lambda^k \rangle.$
\For{each $j \in\ [1:n]$}
\State Initialize $temp \leftarrow $   0
\For{each edge $u\rightarrow v\ \in\ \mathcal G$}
\If{$s(u{\rightarrow} v)=$ `1'}
\State $temp\ +=  \langle p^{for}_u(\lambda),\lambda^{j-1} \rangle \langle p^{rev}_v(\lambda),\lambda^{k-j} \rangle$
\EndIf
\EndFor
\State \textbf{assign} $\sum\limits_{\substack{w| |w|=k,\\w_j=1}} \langle y^1 \uparrow ... \uparrow y^t,w \rangle \gets temp$
\EndFor
\EndFor
\For{$i \in\ [1:n]$}
\State Use \eqref{eq:posterior_prob} to compute $\Pr(X_i=1|Y^1=y^1,...,Y^t=y^t)$
\State $\hat X_i \leftarrow 1$ if $\Pr(X_i=1|Y^1=y^1,...,Y^t=y^t) > 0.5$ and $\hat X_i \leftarrow 0$ otherwise
\EndFor
\State \textbf{return} $\hat X_1\hat X_2...\hat X_n$
\end{algorithmic}
\end{algorithm}

%% file: Sections_revised/numerical.tex
\section{Numerical results}
\label{sec:Numerics}

In this section we show numerics supporting our theoretical results. In all of our experiments, we generate the input sequence uniformly at random (motivated by average case trace reconstruction), and obtain the $t$ traces by passing the input through a deletion channel (with a deletion probability $\delta$) $t$ times. We then reconstruct the input from the obtained traces and measure how \textit{close} the reconstructed sequence is, to the actual input sequence.

We use two metrics  to measure the performance of the reconstruction algorithms: 1. \textit{Hamming error rate}, which is defined as the average Hamming distance between the actual input and the estimated sequence divided by the length of the input sequence and 2. \textit{Edit error rate}, which is defined as the average edit distance between the actual input and the estimated sequence divided by the length of the input sequence.  The reason for using Hamming error rate is that our goal is to reconstruct a \textit{known-length} sequence, which has been the problem formulation throughout this work.  Moreover, the Hamming error rate is also of special interest to us since the symbolwise MAP is an optimal estimator for minimizing the Hamming error rate (see Appendix~\ref{app:smap_hamming} for a proof). We also use edit error rate as it is a typical metric used in the context of insertion/deletion channels.

{\small
\begin{center}
\begin{tabular}{ |P{3cm}|P{9.5cm}|P{2.5cm}|  }
 \hline
 \multicolumn{3}{|c|}{List of trace reconstruction algorithms compared in this work.} \\
\hline
 \bf Abbreviation & \bf Description & \bf Complexity\\
  \hline
Ind. post. comb. & Independent posterior combination (Alg.~\ref{alg:ind_comb}) & $O(n^2t)$\\
 \hline 
 BMA    & Bitwise majority alignment of \cite{Batu2004} (Alg.~\ref{alg:bitwise_majority}) & $O(nt)$ \\
\hline
 Trace stats. & Algorithm based on trace symbolwise statistics from \cite{Holenstein2008} (Alg.~\ref{alg:trace_statistics}) 
 & $O(n^{3.37}+nt)$\\
 \hline
 Grad asc.   &  Projected gradient ascent (Alg.~\ref{alg:grad_asc_traces}) &  $O(n^2t)$\\
  \hline
 SMAP seq. & Sequential symbolwise MAP heuristic (Alg.~\ref{alg:apprx_smap}) & $O(n^2t)$ \\
  \hline
 SMAP exact & Exact symbolwise MAP (Alg.~\ref{alg:exact_smap}) & $O(n^{t+2}2^t)$\\
  \hline
\end{tabular}
\end{center}}

\vspace{5mm}

\noindent \textbf{Baseline algorithms:}

\begin{enumerate}[leftmargin = *]
\item \textbf{Independent posterior combination:} As pointed in the introduction, computing the posterior probabilities for  each deletion channel and combining them as if they came from independent observations does not provide a natural solution for computing the posterior probabilities for the $t$-trace deletion channel. One could, however, check how such a naive combination of posteriors compares with our reconstruction algorithms for $t$-traces. This is detailed as Alg.~\ref{alg:ind_comb}. The complexity of this algorithm is $O(n^2t)$ since computing the posteriors takes $O(n^2)$ and we compute posteriors for $t$ traces.
\item \textbf{Bitwise Majority Alignment (introduced in \cite{Batu2004}):} BMA reconstructs the input sequence by first ``aligning'' the traces using a pointer for each trace, and then taking the majority of the pointed symbols. BMA is detailed as Alg.~\ref{alg:bitwise_majority}. From an efficiency standpoint, BMA is the most efficient of all the algorithms since it is linear in the blocklength as well as the number of traces ($O(nt)$).
\item \textbf{Trace statistics algorithm:} An algorithm based on trace symbol statistics (also called mean-based algorithms and summary statistics algorithms) has been extensively studied for worst-case trace reconstruction (see \cite{Holenstein2008}, \cite{De2017}, \cite{Nazarov:2017}). In essence, the algorithm first estimates the ``trace symbol statistics'' -- $\Pr(Y_i=1)\ \forall\ i$ -- from the obtained traces and uses only these estimates to reconstruct $X$. However, it uses a new set of traces for every position $i$, thus requiring at least $n$ traces (see (3.6) and the paragraph below (3.8) in \cite{Holenstein2008}). Here we modify the algorithm to adapt them for an arbitrary number of traces; in particular, we reuse the traces while estimating  $\Pr(Y_i=1)\ \forall\ i$. The algorithm is detailed in Alg.~\ref{alg:trace_statistics}. 

The complexity analysis for this gets tricky since it depends on the algorithm used to solve the set of $2n$ linear programs. The state-of-the-art algorithm for solving a linear program in $n$ variables takes approximately $O(n^{2.37})$ (see \cite{cohen2019solving}); thus the complexity of Trace statistics algorithm is $O(n^{3.37}+nt)$, where the $nt$ term corresponds  to the complexity of computing $\hat p_j$. However, in our implementation we use the solver from the "SciPy" Python library which uses primal-dual interior point methods for solving linear programs. The complexity of such methods is typically $O(n^3)$ making our implementation $O(n^4+nt)$. Also note that these are iterative methods and have many hidden constants (such as the number of iterations for convergence).
\end{enumerate}
We note that the state-of-the-art average-case trace reconstruction algorithms in the literature are applicable in the asymptotic regime where the blocklength $n$ and the number of traces $t$ approach $\infty$; it is not clear how to adapt such algorithms for a finite blocklength and a small number of traces. It is for this reason that we chose to compare against BMA and Trace statistics algorithm, which can be easily adapted for the finite blocklength regime and for a small number of traces. It should also be noted that the performance of the above two algorithms may not be reliable with a small number of traces (as they are not designed for this regime), yet we include them owing to the lack of better baselines.\\

\begin{algorithm}[t!]
\caption{Trace reconstruction via independent posterior combination}\label{alg:ind_comb}
\begin{algorithmic}[1]
\item Input: Traces {$Y^{1}=y^1,...,Y^{t}=y^t$}, input length $n$ \\ Outputs: Estimate of the input $\hat X$
\State Initialize priors $p^{old} \gets (0.5,0.5,...,0.5)$
\For {$l=1:t$}
\State Use Alg.~\ref{alg:apprx_smap_dp} with $p^{old}$ and $y^l$ to compute posteriors $p^{l,new}$
\EndFor
\For {$i=1:n$}
\If {$\prod_{l=1}^t p^{l,new}_i \geq \prod_{l=1}^t (1-p^{l,new}_i)$}
 $\ \hat X_i \gets 1$
\Else
 $\ \hat X_i \gets 0$
\EndIf
\EndFor
\end{algorithmic}
\end{algorithm}

\begin{algorithm}[!t]
\caption{Bitwise Majority Alignment} 
\label{alg:bitwise_majority}
\begin{algorithmic}[1]
\item Input: Traces {$Y^{1}=y^1,...,Y^{t}=y^t$, input length $n$}\\ Output: estimate of input $\hat X = \hat X_1 \hat X_2...\hat X_n$.
\State Initialize $c_j=1$ for $j\in [t]$.
\State Initialize $\hat X_i = 1$ for $i \in [n]$.
\For{$i \in\ [1:n]$}
\State Let $b$ be the majority over all $t$ of $y^j_{c_j}$
\State $\hat X_i \gets b$
\State Increment $c_j$ for each $j$ such that $y^j_{c_j} = b$ 
\EndFor
\end{algorithmic}
\end{algorithm}

\begin{algorithm}[!t]
\caption{Trace statistics heuristic} 
\label{alg:trace_statistics}
\begin{algorithmic}[1]
\item Input: Traces {$Y^{1}=y^1,...,Y^{t}=y^t$, input length $n$}\\ Output: estimate of input $\hat X = \hat X_1 \hat X_2...\hat X_n$.
\State Append each trace $y^j$ with zeros until each of them is of length $n$.
\State Assign $\hat p_j \leftarrow \frac{|\{y^l:y^l_j=1\}|}{t}$.
\For{$i \in\ [1:n]$}
\State Solve the 2 linear programs (3.6) in \cite{Holenstein2008}  by fixing $x_i=0$ and $x_i=1$: let the optimum value in the two cases be $m_0$ and $m_1$ respectively.
\State If $m_0<m_1$, assign $\hat X_i = x_i \leftarrow 0$. Else fix $\hat X_i = x_i \leftarrow 1$. 
\EndFor
\end{algorithmic}
\end{algorithm}

\newpage
\noindent \textbf{Algorithms introduced in this paper:}

\begin{enumerate}[leftmargin = *]
\item \textbf{Projected gradient ascent:} Alg.~\ref{alg:grad_asc_traces} used as described, with max iterations $M=100$ and convergence criteria $C$ set as follows: the percentage difference in $\sum_j \mathbf F(p,y^j)$ over two consecutive iterations is less than 0.1\%.
\item \textbf{Symbolwise MAP sequentially used one trace at a time:} Alg.~\ref{alg:apprx_smap} used as described.
\item \textbf{Exact symbolwise MAP:} Alg.~\ref{alg:exact_smap} used as described.
\end{enumerate}

\vspace{7mm}

\textbf{Observations:} In Fig.~\ref{fig:numerics_hamming} and Fig.~\ref{fig:numerics_edit}, we compare the Hamming and edit error rates for the different algorithms described above.
\begin{itemize}[leftmargin = *]
\item The 3 algorithms introduced in this work outperform the 3 baselines in most cases. The Hamming error rate of Grad asc. with 2 and 3 traces is a notable exception as it does worse than Ind. post. comb. However, it improves rapidly as we increase the number of traces as seen in Fig.~\ref{fig:numerics_hamming}. 
\item Both Ind. post. comb. as well as our SMAP seq. struggle with the problem of \textit{diminishing returns} for Hamming error rate as they do not improve much with the number of traces. This could indicate that considering traces one at a time could fail to accumulate extrinsic information (for instance, it completely neglects the possible alignments given multiple traces); one needs to simultaneously consider multiple traces in order to accomplish this. SMAP seq. however, improves with the number of traces with respect to edit error rate.
\item The Grad asc. is the ``champion'' amongst the algorithms we compare here, when it comes to the edit error rate as illustrated by Fig.~\ref{fig:numerics_edit}. The Grad asc. was constructed with the aim of maximizing the likelihood of the observed traces, and this in turn seems to have some correlation with minimizing the edit distance -- it is not clear why this is the case. 
\item As seen in Fig.~\ref{fig:numerics_hamming} (a) and (b), SMAP exact has the minimum Hamming error rate. This supports the fact that symbolwise MAP is the minimizer of the Hamming error rate. However, note that this does not necessarily minimize the edit error rate, as seen from Fig.~\ref{fig:numerics_edit} (a) and (b). 
\end{itemize}

\begin{figure}[!t]
\centering
\includegraphics[scale=0.85]{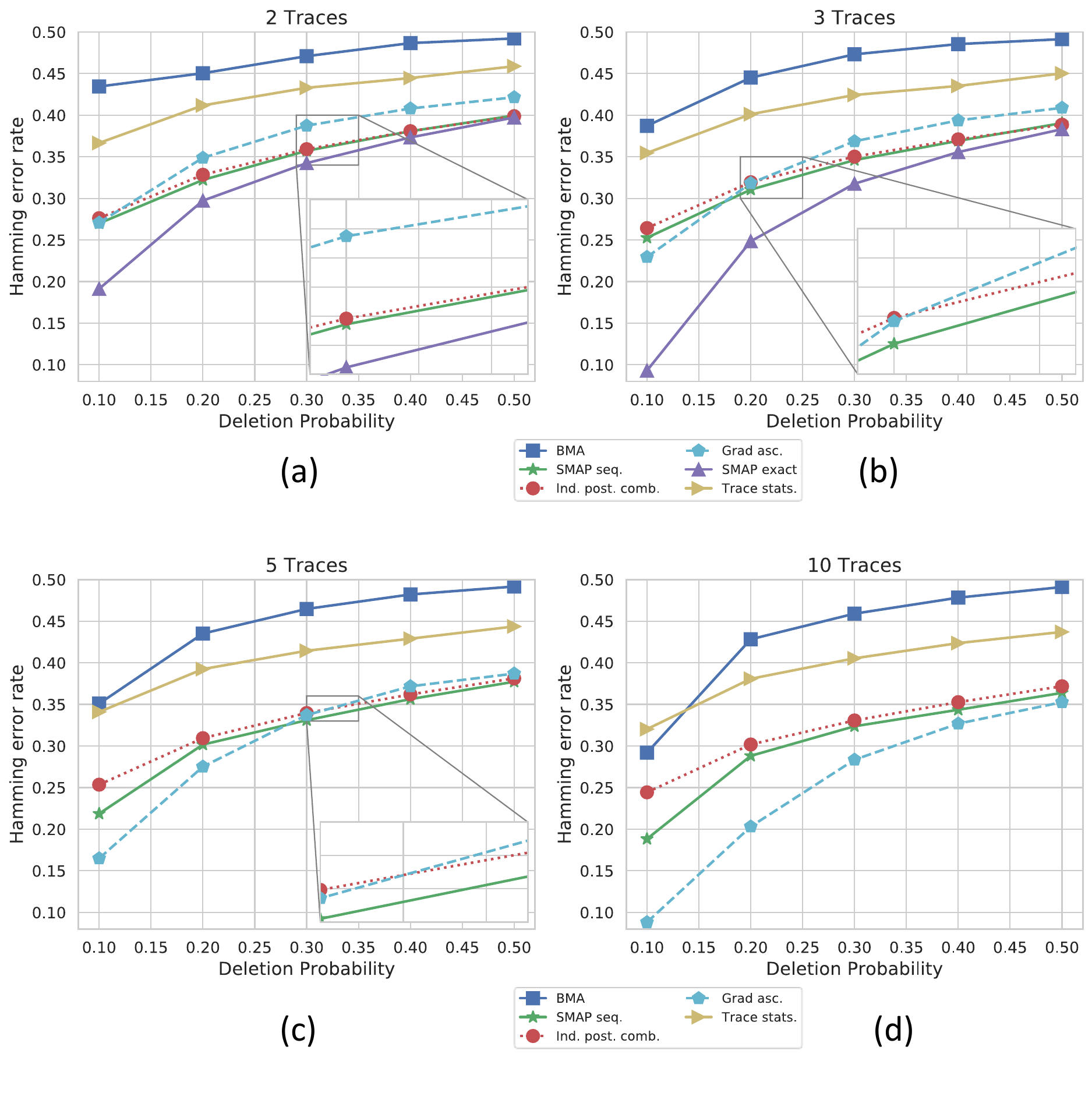}
\caption{Comparison of Hamming error rates for a blocklength $n=100$ illustrated with 2,3,5 and 10 observed traces. Note that we do not run SMAP exact. for 5 and 10 traces since its complexity grows exponentially with the number of traces. All the subplots are plotted on the same scale to aid comparability across subplots. Few of the subplots which contain algorithms with similar error rates also contain a zoomed-in inset view.}
\label{fig:numerics_hamming}
\end{figure}
\clearpage

\begin{figure}[!t]
\centering
\includegraphics[scale=0.85]{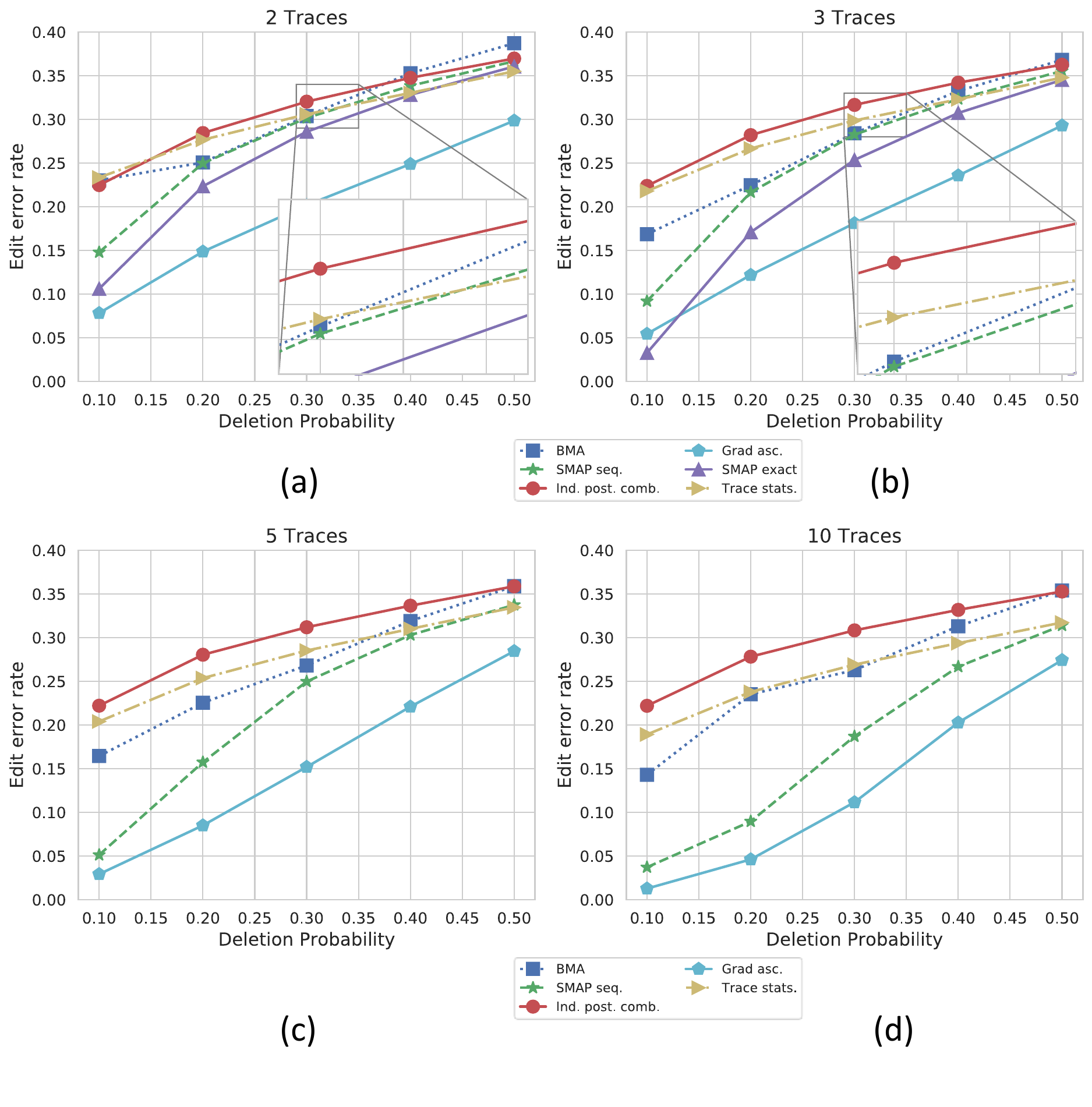}
\caption{Comparison of edit error rates for a blocklength $n=100$ illustrated with 2,3,5 and 10 observed traces. Note that we do not run SMAP exact. for 5 and 10 traces since its complexity grows exponentially with the number of traces. All the subplots are plotted on the same scale to aid comparability across subplots. Few of the subplots which contain algorithms with similar error rates also contain a zoomed-in inset view.}
\label{fig:numerics_edit}
\end{figure}
\clearpage

%% file: Sections_revised/conclusions.tex
\section{Conclusions}

\label{sec:conclusions}
In this work we gave, to the best of our knowledge, the first results and techniques to compute posterior distributions over single and multiple deletion channels. We also provided a new perspective on the maximum-likelihood for the deletion channel by showing an equivalence between a discrete optimization problem and its relaxed version. In this process, we introduced a variety of tools (the relaxed binomial coefficient, edit graph and infiltration product) and demonstrated their use for analyzing deletion channels.  We also presented numerical evaluations of our algorithms and showed performance improvements over existing trace reconstruction algorithms.


%% file: Sections_revised/Appendix.tex
\subsection{Proofs.}

\subsubsection{Proof of Theorem~\ref{thm:channel_equiv}}
\label{app:channel_equiv}

The intuition behind the theorem is that the cascade model splits the error events in the $t$-trace deletion channel into 2 parts:\\
- When an input symbol is deleted in all the traces, which is captured by the deletion channel with parameter $\delta^t$.\\
- When an input symbol is not deleted in at least one of the traces, captured by the remnant channel.
\begin{figure}[!h]
\centering
\includegraphics[scale=0.4]{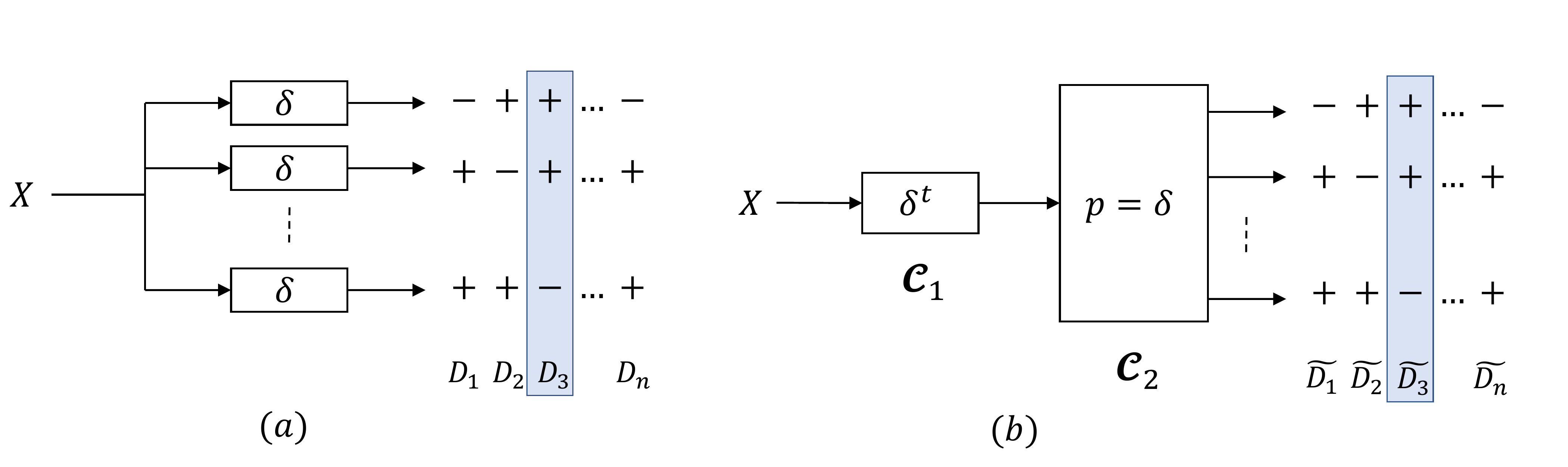}
\caption{The deletion error events occurring in the two channel models. Here `$-$' corresponds to a symbol being deleted and `$+$' corresponds to a transmission. The deletion pattern $D_i$ corresponds to the input symbol $X_i$.}
\end{figure}

In order to prove the theorem, we need to prove that the deletion patterns arising in the $t$-trace channel model and in the cascade model have the same distribution, i.e.,
\begin{align*}
\Pr(D_1=d_1,D_2=d_2,...,D_n=d_n) = \Pr(\widetilde D_1=d_1,\widetilde D_2=d_2,...,\widetilde D_n=d_n),
\end{align*}
where $d_i \in \{-,+\}^t$, where a $-$ corresponds to a deletion and a $+$ corresponds to a transmission.
Also from the definition of our channel models, the deletions act independently on each input symbol i.e., $D_i\ \indep\ D_j$ for $i\neq j$. So it is sufficient to prove that the distributions of each $D_i$ and $\widetilde D_i$ are the same.

\begin{figure}[!h]
\centering
\includegraphics[scale=0.5]{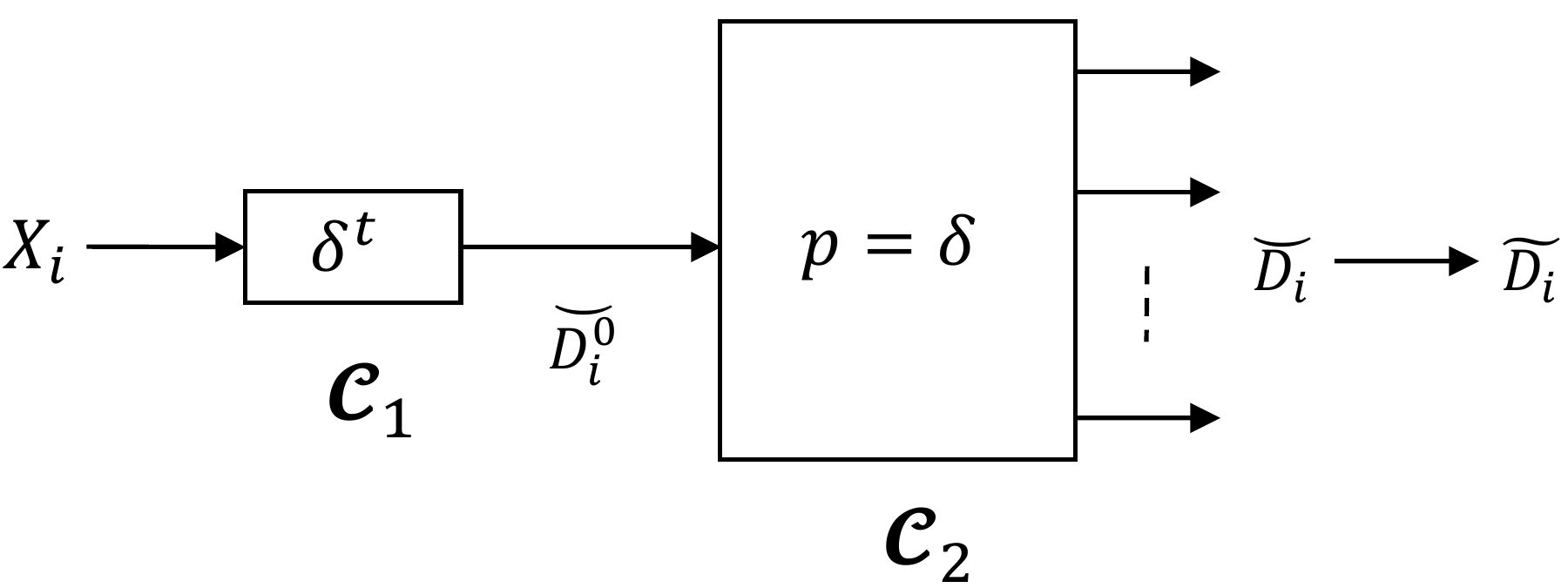}
\caption{The error events of the cascade model, expressed in terms of the error events of its components.}
\end{figure}

Consider $\widetilde D_i$ -- this is influenced by $\breve {D^0_i}$ which is the deletion in channel $\mathcal C_1$ and by $\breve D_i$ which are the deletion in the remnant channel $\mathcal C_2$. To prove the equivalence, we consider 2 cases:
\begin{itemize}
\item $d_i = (-,-,-,...,-)$, the error event where a symbol is deleted in all the observations. It can be seen that $\Pr(D_i = d_i)$ for this case is $\delta^t$. On the other hand, to compute $\Pr(\widetilde D_i = d_i)$, we note that this event is possible if and only if $\breve {D^0_i} = -$, since by definition, the remnant channel cannot delete the input symbol in all the $t$ observations. Therefore, $\Pr(\widetilde D_i = d_i)=\Pr(\breve D^0_i = -) = \delta^t$.
\item $d_i \neq (-,-,-,...,-)$, i.e., the input symbol is not deleted in at least one trace. Also let us define $k$ to be the count of $-$ in $d_i$. In this case, $\Pr(D_i = d_i)=\delta^{\text{Count(-) in } d_i}(1-\delta)^{\text{Count(+) in } d_i} = \delta^k (1-\delta)^{t-k}$. For the cascade model, this event requires that 
$\breve {D^0_i} = +$ and $\breve {D_i}=d_i$. Thus, $$\Pr(\tilde D_i=d_i)=\Pr(\breve {D^0_i} = +)\cdot \Pr(\breve {D_i}=d_i) = (1-\delta^t) \frac{\delta^k (1-\delta)^{t-k}}{1-\delta^t}=\delta^k (1-\delta)^{t-k}.$$
\end{itemize}
In both cases, the distributions of $D_i$ and $\widetilde D_i$ are the same, proving the equivalence.

\subsubsection{Proof of Lemma~\ref{lemma:F_decomposition}}
\label{app:F_lemma_proof}
\deletionMLrellemma*
\begin{proof}
The proof of this lemma uses a similar approach as the proof of Thm.~\ref{thm:approx_smap_postprob}. First, in the expression for $\mathbf F(\cdot )$, we separate out the subsets that contain index $i$:
\begin{align*}
\mathbf F(p,y) &= \sum\limits_{\substack{\mathcal S|\mathcal S\subseteq [n],\\|\mathcal S|=m}} \quad \prod\limits_{j=1}^m p_{\mathcal S_j}^{y_j} (1-p_{\mathcal S_j})^{1-y_j}\\
&= \sum\limits_{\substack{\mathcal S|\mathcal S\subseteq [n],\\|\mathcal S|=m, \\ i \notin \mathcal S}} \quad \prod\limits_{j=1}^m p_{\mathcal S_j}^{y_j} (1-p_{\mathcal S_j})^{1-y_j} + \sum\limits_{\substack{\mathcal S|\mathcal S\subseteq [n],\\|\mathcal S|=m,\\ i \in \mathcal S}} \quad \prod\limits_{j=1}^m p_{\mathcal S_j}^{y_j} (1-p_{\mathcal S_j})^{1-y_j}\\
&= \mathbf F(p_{[n]\backslash \{i\}},y) + \sum\limits_{\substack{\mathcal S|\mathcal S\subseteq [n],\\|\mathcal S|=m,\\ i \in \mathcal S}} \quad \prod\limits_{j=1}^m p_{\mathcal S_j}^{y_j} (1-p_{\mathcal S_j})^{1-y_j}. \numberthis
\label{eq:Flemma_proof1}
\end{align*}
Now the second term can be further split as,
\begin{align*}
\sum\limits_{\substack{\mathcal S|\mathcal S\subseteq [n],\\|\mathcal S|=m,\\ i \in \mathcal S}} \quad \prod\limits_{j=1}^m p_{\mathcal S_j}^{y_j} (1-p_{\mathcal S_j})^{1-y_j} &= \sum_{k=1}^m\sum\limits_{\substack{\mathcal S|\mathcal S\subseteq [n],\\|\mathcal S|=m,\\ \mathcal S_k = i}} \quad \prod\limits_{j=1}^m p_{\mathcal S_j}^{y_j} (1-p_{\mathcal S_j})^{1-y_j}.
\end{align*}
One could express the set $\mathcal S$ as the union  $\mathcal S = \mathcal S' \cup \{i\} \cup \mathcal S''$ such that $\mathcal S' \subseteq [i-1]$ and $\mathcal S'' \subseteq [i+1:n]$ to get
\begin{align*}
&\sum_{k=1}^m\sum\limits_{\substack{\mathcal S|\mathcal S\subseteq [n],\\|\mathcal S|=m,\\ \mathcal S_k = i}} \quad \prod\limits_{j=1}^m p_{\mathcal S_j}^{y_j} (1-p_{\mathcal S_j})^{1-y_j}\\
&= \sum_{k=1}^m\sum_{\substack{\mathcal S'|\\ \mathcal S'\subseteq [i-1]\\ |\mathcal S'|=k-1}}\sum_{\substack{\mathcal S''|\\ \mathcal S'' \subseteq [i+1:n]\\ |\mathcal S''|=m-k}}  \left( \prod\limits_{j=1}^{k-1} p_{\mathcal S'_j}^{y_j} (1-p_{\mathcal S'_j})^{1-y_j} \right) \left( p_{i}^{y_k} (1-p_{i})^{1-y_k} \right) \left( \prod\limits_{j=1}^{m-k} p_{\mathcal S''_j}^{y_{j+k}} (1-p_{\mathcal S''_j})^{1-y_{j+k}} \right)\\
&= \sum_{k=1}^m  p_{i}^{y_k} (1-p_{i})^{1-y_k}\left(\sum_{\substack{\mathcal S'|\\ \mathcal S'\subseteq [i-1]\\ |\mathcal S'|=k-1}} \prod\limits_{j=1}^{k-1} p_{\mathcal S'_j}^{y_j} (1-p_{\mathcal S'_j})^{1-y_j} \right)  \left( \sum_{\substack{\mathcal S''|\\ \mathcal S'' \subseteq [i+1:n]\\ |\mathcal S''|=m-k}}  \prod\limits_{j=1}^{m-k} p_{\mathcal S''_j}^{y_{j+k}} (1-p_{\mathcal S''_j})^{1-y_{j+k}} \right)\\
&= \sum_{k=1}^m  p_{i}^{y_k} (1-p_{i})^{1-y_k} \mathbf F(p_{[i-1]},y_{[k-1]})   \mathbf F(p_{[i+1:n]},y_{[k+1:m]}).
\end{align*}
The $\sum_{k=1}^m$ summation in the above expression could further be split into the two cases depending on whether $y_k=0$ or $y_k=1$, which simplifies the term  $p_{i}^{y_k} (1-p_{i})^{1-y_k}$ to either $1-p_i$ or $p_i$ respectively. Thus,
\begin{align*}
&\sum\limits_{\substack{\mathcal S|\mathcal S\subseteq [n],\\|\mathcal S|=m,\\ i \in \mathcal S}} \quad \prod\limits_{j=1}^m p_{\mathcal S_j}^{y_j} (1-p_{\mathcal S_j})^{1-y_j}\\
&= (1-p_i)\sum_{k|y_k=0} \mathbf F(p_{[i-1]},y_{[k-1]})   \mathbf F(p_{[i+1:n]},y_{[k+1:m]}) + p_i\sum_{k|y_k=1} \mathbf F(p_{[i-1]},y_{[k-1]})   \mathbf F(p_{[i+1:n]},y_{[k+1:m]}).\numberthis
\label{eq:Flemma_proof2}
\end{align*}
\noindent Plugging \eqref{eq:Flemma_proof2} in \eqref{eq:Flemma_proof1} concludes the proof of the Lemma.
\end{proof}

\subsubsection{Proof of Lemma~\ref{lemma:bin_inf_relation}}
The following Lemma forms the backbone of  the analyses for multiple traces. This lemma is also closely related to the channel equivalence in Theorem~\ref{thm:channel_equiv}.
\label{app:bin_inf_lemma}
\bininfrelation*

\begin{proof}
The channel equivalence can essentially be tied to this lemma as follows: consider the two channel models in Fig.~\ref{fig:channel_equiv}. The probability of observations given the input in both cases is proportional to the number of ways of obtaining the observations given the input.
\begin{itemize}
\item For the $t$-trace deletion channel model in Fig.~\ref{fig:channel_equiv} (a), the number of ways to obtain the traces given the input is equal to ${X \choose {Y}^{1}}{X \choose {Y}^{2}}...{X \choose {Y}^{t}}$.
\item For the cascade model in Fig.~\ref{fig:channel_equiv} (b), the number of ways to obtain the traces given the input is equal to $\sum_{z} {X \choose z} \langle \tilde  Y^{1}\uparrow \tilde Y^{2} \uparrow...\uparrow \tilde  Y^{t},z \rangle$, which we show below.
\end{itemize}
The above two are expression must be equal since the two channel models are equivalent.

We now first compute the probability of a given set of output sequences given an input sequence for the remnant channel, namely $\Pr(\tilde  Y^{1},\tilde Y^{2},...,\tilde  Y^{t}|Z)$. First, note that there can be multiple deletion patterns corresponding to outputs $\tilde  Y^{1},\tilde Y^{2},...,\tilde  Y^{t}$ resulting from a given input $Z$. The number of such patterns is equal to $\langle \tilde  Y^{1}\uparrow \tilde Y^{2} \uparrow...\uparrow \tilde  Y^{t},Z \rangle$, which essentially follows from the definition of the infiltration product. Consider one such valid deletion pattern, i.e., a deletion pattern $\mathcal{D}$ that is a mapping of the symbols in $Z$ onto the symbols in $\tilde  Y^{1},\tilde Y^{2},...,\tilde  Y^{t}$: $\mathcal{D}=\{(1,S_1),(2,S_2),...,(|Z|,S_{|Z|})\}$. Here $(i,S_i)$ represents the fact that $Z_i$ is not deleted in the output set $\tilde{Y}^{S_i}$ and is deleted in the rest. From the definition of the remnant channel, we have $|S_i|>0$ . Also $\sum_{i=1}^{|Z|}|S_i|=\sum_{j=1}^t|\tilde{Y}^{j}|$ since every symbol of each output is associated with exactly one input symbol and hence corresponds to one particular $S_i$. Thus,
\begin{align*}
\Pr(\tilde  Y^{1},\tilde Y^{2},...,\tilde  Y^{t}|Z) &= \langle \tilde  Y^{1}\uparrow \tilde Y^{2} \uparrow...\uparrow \tilde  Y^{t},Z \rangle \Pr(\tilde  Y^{1},\tilde Y^{2},...,\tilde  Y^{t}|Z,\mathcal{D})\\
&=\langle \tilde  Y^{1}\uparrow \tilde Y^{2} \uparrow...\uparrow \tilde  Y^{t},Z \rangle \prod_{i=1}^{|Z|}\frac{(1-\delta)^{|S_i|}\delta^{t-|S_i|}}{1-\delta^t}\\
&=\langle \tilde  Y^{1}\uparrow \tilde Y^{2} \uparrow...\uparrow \tilde  Y^{t},Z \rangle \frac{(1-\delta)^{\sum|S_i|}\delta^{|Z|t-\sum |S_i|}}{(1-\delta^t)^{|Z|}}\\
&=\langle \tilde  Y^{1}\uparrow \tilde Y^{2} \uparrow...\uparrow \tilde  Y^{t},Z \rangle \frac{(1-\delta)^{\sum|\tilde{Y}^{j}|}\delta^{|Z|t-\sum |\tilde{Y}^{j}|}}{(1-\delta^t)^{|Z|}}.
\end{align*}
We can then compute the probability of the output given the input for the cascade channel as 
\begin{align*}
\Pr(&\tilde Y^{1},\tilde Y^{2},...,\tilde  Y^{t}|X)\\ &= \sum_{z} \Pr(\tilde Y^{1},\tilde Y^{2},...,\tilde  Y^{t},Z=z|X)\\
&= \sum_{z} \Pr(Z=z|X)\Pr(\tilde Y^{1},\tilde Y^{2},...,\tilde  Y^{t}|Z=z)\\
&= \sum_{z} \Bigg[{X \choose z} \delta^{t(|X|-|z|)}(1-\delta^t)^{|z|}\langle \tilde  Y^{1}\uparrow \tilde Y^{2} \uparrow...\uparrow \tilde  Y^{t},z \rangle \frac{(1-\delta)^{\sum|\tilde{Y}^{j}|}\delta^{|z|t-\sum |\tilde{Y}^{j}|}}{(1-\delta^t)^{|z|}}\Bigg]\\
&= \left[\sum_{z} {X \choose z} \langle \tilde  Y^{1}\uparrow \tilde Y^{2} \uparrow...\uparrow \tilde  Y^{t},z \rangle\right] \delta^{t|X|-\sum |\tilde{Y}^{j}|} {(1-\delta)^{\sum|\tilde{Y}^{j}|}}. \numberthis
\label{eq:lemma2_proof_1}
\end{align*}

For the $t$-trace deletion channel model, we have:
\begin{align*}
\Pr(Y^1,Y^2,...,Y^t|X) &= \prod_{j=1}^t {X \choose Y^j} \delta^{|X|-|Y^j|}(1-\delta)^{|Y_j|} \\
 &= {X \choose {Y}^{1}}{X \choose {Y}^{2}}...{X \choose {Y}^{t}} \delta^{t|X|-\sum |{Y}^{j}|} {(1-\delta)^{\sum|{Y}^{j}|}}. \numberthis
 \label{eq:lemma2_proof_2}
 \end{align*}
Equating \eqref{eq:lemma2_proof_1} and \eqref{eq:lemma2_proof_2} with $X = h$ and traces as $Y^j = \tilde Y^j = f_j$ proves the Lemma.

Alternatively, we use also induction to prove the statement as we do below. The statement is trivially true when $m=1$ since, $\sum_{w}{h \choose w}\langle f_1,w \rangle={h \choose f_1}$ as $\langle f,w \rangle=\mathbbm{1}_{f=w}$. We refer the reader to equation 6.3.25 in \cite{lothaire1997combinatorics} for the proof of the lemma for the case $m=2$. Assume that the statement is true for $m=k \in \mathbb{Z}, k\geq 2$. We next prove the validity when $m=k+1$. \\
Consider 
\begin{align}
{h \choose f_1} {h \choose f_2}...{h \choose f_k}{h \choose f_{k+1}}&=\sum_{w}{h \choose w}\langle f_1\uparrow f_2\uparrow  ...\uparrow f_k,w \rangle {h \choose f_{k+1}}\nonumber\\
&=\sum_{w}\left[{h \choose w} {h \choose f_{k+1}}\right]\langle f_1\uparrow f_2\uparrow  ...\uparrow f_k,w \rangle\nonumber\\
&=\sum_{w}\left[\sum_v \langle w\uparrow f_{k+1},v \rangle {h \choose v}\right]\langle f_1\uparrow f_2\uparrow  ...\uparrow f_k,w \rangle\nonumber\\
&=\sum_{v}{h \choose v}\left[\sum_w \langle w\uparrow f_{k+1},v \rangle \langle f_1\uparrow f_2\uparrow  ...\uparrow f_k,w \rangle\right]. \label{eq:prop2proof}
\end{align} 
To evaluate the term in the square bracket, we use \eqref{def:infiltforseq}. For the case where $\tau \in \mathcal{A}^*,\sigma \in \mathbb{Z}\langle \mathcal{A} \rangle$ in \eqref{def:infiltforseq}, we have $$\sigma\uparrow \tau=\sum_{f\in \mathcal{A}^*} \langle \sigma,f \rangle  (f\uparrow \tau),$$ and thus \begin{equation}
\langle \sigma \uparrow \tau,u \rangle=\sum_{f\in \mathcal{A}^*} \langle \sigma,f \rangle  \langle f\uparrow \tau,u\rangle.
\label{eq:prop2proof2}
\end{equation}  We use \eqref{eq:prop2proof2}  to replace the term in the square bracket in \eqref{eq:prop2proof}, i.e.,
\begin{align}
{h \choose f_1} {h \choose f_2}...{h \choose f_k}{h \choose f_{k+1}}\nonumber\\
=\sum_{v}{h \choose v}\langle (f_1\uparrow f_2\uparrow  ...\uparrow f_k) \uparrow f_{k+1},v \rangle,
\end{align} 
and the lemma follows from the associativity property of the infiltration product.
\end{proof}

\subsubsection{Proof of Lemma~\ref{lemma:smapsum}}
\label{app:smapsum}
\begin{restatable}{lemma}{smapsum}
\label{lemma:smapsum}
\begin{align*}
\sum_{\substack{f||f|{=}n\\f_i{=}a}}{f \choose g} = 2^{n-|g|}\Bigg(\frac{1}{2}{n-1 \choose |g|} &+\sum_{j|g_j=a}{i-1 \choose j-1}{n-i \choose |g|-j}\Bigg),
\end{align*}
where $j\in\Big[\max\{1,|g|+i-n\}:\min\{i,|g|\}\Big]$. 
\end{restatable}
\begin{proof}
First, observe that $${f \choose g} = \sum_{\substack{S\subseteq [n]:\\|S|=|g|}} \mathbbm 1_{f_S=g},$$ where the summation is over all ordered subsets of $[n]=\{1,2,...,n\}$ of size $|g|$ and $f_S$ corresponds to the subsequence of $f$ indexed by $S$. Thus, 
\begin{align*}
\sum_{\substack{f\in\mathcal{A}^n|\\f_i=a}}&{f \choose g} = \sum_{\substack{f\in\mathcal{A}^n|\\f_i=a}} \sum_{\substack{S\subseteq [n]|\\|S|=|g|}} \mathbbm 1_{f_S=g} = \sum_{\substack{S\subseteq [n]|\\|S|=|g|}} \sum_{\substack{f\in\mathcal{A}^n|\\f_i=a}} \mathbbm 1_{f_S=g}\\
&=\sum_{\substack{S\subseteq [n]|\\|S|=|g|\\i \notin S}} \sum_{\substack{f\in\mathcal{A}^n|\\f_i=a}} \mathbbm 1_{f_S=g} + 
\sum_{\substack{S\subseteq [n]|\\|S|=|g|\\i \in S}} \sum_{\substack{f\in\mathcal{A}^n|\\f_i=a}} \mathbbm 1_{f_S=g}\\
&=\sum_{\substack{S\subseteq [n]|\\|S|=|g|\\i \notin S}} \sum_{\substack{f\in\mathcal{A}^n|\\f_i=a}} \mathbbm 1_{f_S=g} + 
\sum_{j=1}^m \sum_{\substack{S\subseteq [n]|\\|S|=|g|\\S_j=i}} \sum_{\substack{f\in\mathcal{A}^n|\\f_i=a}} \mathbbm 1_{f_S=g}. \numberthis \label{eq:lemmasmapsum}
\end{align*}
\begin{figure}[!h]
\begin{center}
\includegraphics[scale=0.25]{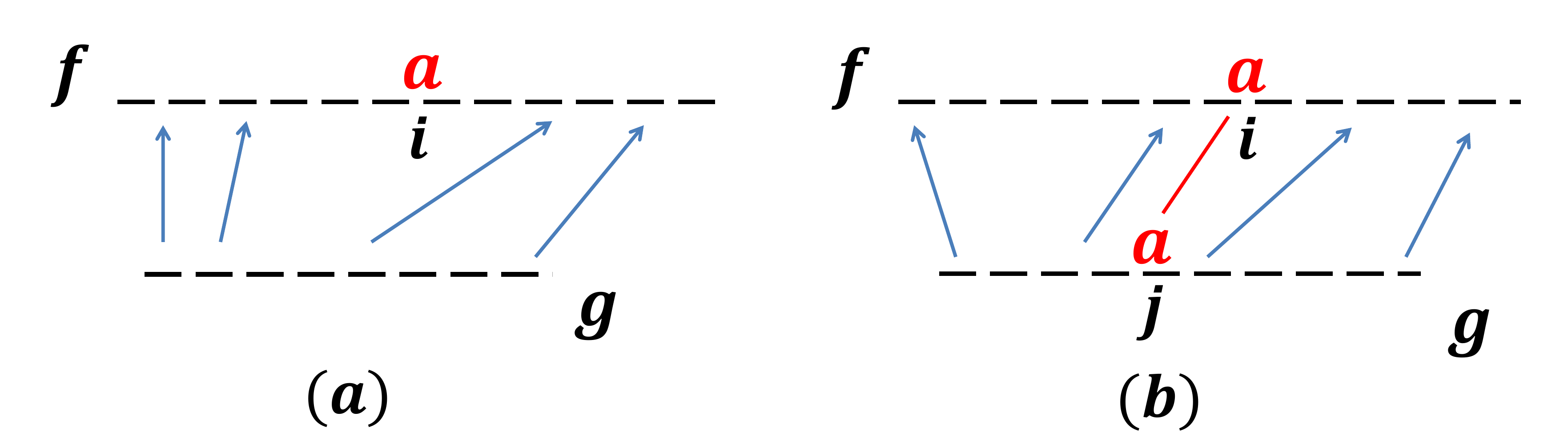}
\vspace*{-2mm}
\caption{Figure illustrating proof of Lemma~\ref{lemma:smapsum}.}
\label{fig:lemmasmapsum}
\end{center}
\end{figure}
The two terms in \eqref{eq:lemmasmapsum} can be visualized as the number of ways to fill up the blank spaces (spaces without arrows pointing to it in $f$) in Fig.~\ref{fig:lemmasmapsum}(a) and (b) respectively. Solving this counting problem, we get 
$$\sum_{\substack{f\in\mathcal{A}^n|\\f_i=a}}{f \choose g}=2^{n-|g|}\left(\frac{1}{2}{n-1 \choose |g|}+\sum_{j|g_j=a}{i-1 \choose j-1}{n-i \choose |g|-j}\right).$$
\end{proof}

\subsection{Dynamic program to compute $\mathbf F(\cdot)$ and $\nabla \mathbf F(\cdot)$}
\subsubsection{Computation of $\mathbf F(p,v)$}
\label{app:F_compute}
We here describe how to compute $\mathbf F(p,v)$ in $O(mn)$ time and space complexity, where $p=(p_1,...,p_n)$ and $v=v_1...v_m$, via a dynamic programming approach. Note that $m \leq n$ otherwise $\mathbf F(p,v)=0$.
We first define
\begin{align*}
  &\mathbf G^{for}(k,j)\triangleq \mathbf F(p_{[1:k]},v_{[1:j]}).
  \numberthis 
\label{eq:G}
  \end{align*}
  
Using Lemma~\ref{lemma:F_decomposition} with $i=n$, we get
\begin{align*}
\mathbf F(p,v) = \mathbf F(p_{[n-1]},v) + p_n^{v_m} (1-p_n)^{(1-v_m)}  \mathbf F(p_{[n-1]},v_{[m-1]}).
\end{align*}

\noindent This translates to the following dynamic program for $\mathbf G^{for}$:
\begin{align*}
\mathbf G^{for}(k,j) = \mathbf G^{for}(k-1,j)+ p_{k}^{v_j}(1-p_{k})^{1-v_j}&\mathbf G^{for}(k-1,j-1),\numberthis \label{eq:approx_smap_dpfor}
\end{align*}
with the boundary conditions $\mathbf G^{for}(k,0)=1\ \forall\ k \geq 0$ and $\mathbf G^{for}(k,j)=0\ \forall\ k<j$.
The algorithm is now summarized as Alg.~\ref{alg:F_comp}.

\begin{algorithm}[h!]
\caption{Computing $\mathbf F(p,v)$}\label{alg:F_comp}
\begin{algorithmic}[1]
\State Inputs: $p \in [0,1]^n$, $v \in \{0,1\}^m$ 
\State Outputs: $\mathbf F(p_{[1:k]},v_{[1:j]})$ for all $k \in [n]$ and $j\in[m]$ 
\State Initialize $\mathbf G^{for}(k,0)=1\ \forall\ k$ and $\mathbf G^{for}(k,j)=0\ \forall\ k<j$
\For {$k = 1:n$ and $j = 1:m$}
\State Use \eqref{eq:approx_smap_dpfor} to update $\mathbf G^{for}(k,j)$
\EndFor
\State return $\mathbf G^{for}(k,j)\ \forall\ k,j$
\end{algorithmic}
\end{algorithm}

We note that a similar dynamic programming approach yields $\mathbf F(p_{[k+1:n]},v_{[j+1:m]})$ for all $k \in [n]$ and $j\in[m]$ in $O(mn)$ time and space complexity by defining
\begin{align*}
  &\mathbf G^{rev}(k,j)\triangleq \mathbf F(p_{[k+1:n]},v_{[j+1:m]}).
  \end{align*}
\noindent The following dynamic program can be used for $\mathbf G^{rev}$:
\begin{align*}
\mathbf G^{rev}(k,j) = \mathbf G^{rev}(k+1,j)+ p_{k+1}^{v_{j+1}}(1-p_{k+1})^{1-v_{j+1}}&\mathbf G^{rev}(k+1,j+1),\numberthis \label{eq:approx_smap_dprev}
\end{align*}
with the boundary conditions $\mathbf G^{rev}(k,m)=1\ \forall\ k \geq 0$ and $\mathbf G^{rev}(k,j)=0\ \forall\ k,j: n-k<m-j$. \\ 
  
\subsubsection{Computation of $\nabla_p \mathbf F(p,v)$} 
\label{app:F_grad_comp}
First, from Lemma~\ref{lemma:F_decomposition}, we have
\begin{align*}
\mathbf F(p,v) =  \mathbf F(p_{[n]\backslash \{i\}},v) +& (1-p_i)\sum_{k|v_k=0} \mathbf F(p_{[i-1]},v_{[k-1]})   \mathbf F(p_{[i+1:n]},v_{[k+1:m]}) \\&+ p_i\sum_{k|v_k=1} \mathbf F(p_{[i-1]},v_{[k-1]})   \mathbf F(p_{[i+1:n]},v_{[k+1:m]}).
\end{align*}
Differentiating with respect to $p_i$, we get
\begin{align*}
\frac{\partial \mathbf F(p,v)}{\partial p_i} &=  \sum_{k|v_k=1} \mathbf F(p_{[i-1]},v_{[k-1]})   \mathbf F(p_{[i+1:n]},v_{[k+1:m]}) - \sum_{k|v_k=0} \mathbf F(p_{[i-1]},v_{[k-1]})   \mathbf F(p_{[i+1:n]},v_{[k+1:m]}) \\
&=  \sum_{k|v_k=1} \mathbf G^{for}(i{-}1,k{-}1)\mathbf G^{rev}(i,k) -  \sum_{k|v_k=0}\mathbf G^{for}(i{-}1,k{-}1)\mathbf G^{rev}(i,k). \numberthis
\label{eq:F_grad_comp}
\end{align*}
Thus, computing the $\mathbf G^{for}$ and $\mathbf G^{rev}$ terms is sufficient to compute the gradient. As discussed above, this computation requires $O(nm)$ operations. Given $\mathbf G^{for}$ and $\mathbf G^{rev}$, the computation of each partial derivative $\frac{\partial \mathbf F(p,v)}{\partial p_i}$ requires $O(m)$ operations, and we need to compute $n$ such partial derivatives. Thus, the complexity of computing $\nabla_p \mathbf F(p,v)$ can be done in $O(nm)$ time and space complexity.

\begin{algorithm}[h!]
\caption{Computing $\nabla_p \mathbf F(p,v)$}\label{alg:F_grad_comp}
\begin{algorithmic}[1]
\State Inputs: $p \in [0,1]^n$, $v \in \{0,1\}^m$ 
\State Outputs: $\nabla_p \mathbf F(p,v)$
\State Initialize $\mathbf G^{for}(k,0)=1\ \forall\ k$ and $\mathbf G^{for}(k,j)=0\ \forall\ k<j$
\State Initialize $\mathbf G^{rev}(k,m)=1\ \forall\ k$ and $\mathbf G^{rev}(k,j)=0\ \forall\ k,j: n-k<m-j$
\For {$k = 1:n$ and $j = 1:m$}
\State Use \eqref{eq:approx_smap_dpfor} and \eqref{eq:approx_smap_dprev} to compute $\mathbf G^{for}(k,j)$ and $\mathbf G^{rev}(k,j)$
\EndFor
\For {$i = 1:n$}
\State Use \eqref{eq:F_grad_comp} to compute $\frac{\partial \mathbf F(p,v)}{\partial p_i} $
\EndFor
\State return $\nabla_p \mathbf F(p,v)$
\end{algorithmic}
\end{algorithm}

\subsection{An algebraic definition of the infiltration product.}
\label{app:infil_def}
For completeness, we reproduce the formal definition of the infiltration product from Section 6.3 of \cite{lothaire1997combinatorics} (also see there for the equivalence of the two definitions). A \textit{formal series} with indeterminates (or variables) in a set $\mathcal A$ and coefficients in a commutative ring $\mathcal R$, is a mapping of $\mathcal A^*$ onto $\mathcal R$. Recall that a commutative ring is a set which forms an abelian group under an \textit{addition} operation, is a monoid under a \textit{multiplication} operation which commutes, and the multiplication operation distributes over addition. Here we consider $\mathbb Z$, the set of integers as the commutative ring  $\mathcal{R}$. A formal series is called a \textit{polynomial} if only a finite number of sequences are mapped to non-zero values, the rest of the sequences map to zero. Consider two polynomials $\sigma,\tau: \mathcal{A}^*\rightarrow \mathbb Z$. The value taken by a sequence $w\in \mathcal A^*$ on $\sigma$ (or the coefficient of $w$ in $\sigma$) is denoted by $\langle \sigma,w\rangle \in \mathbb R$. We also define binary addition ($\oplus$) and multiplication operations ($\times$) on the set of polynomials as follows:
\begin{align}
\langle \sigma\oplus \tau,w\rangle \triangleq \langle \sigma,w\rangle +  \langle \tau,w \rangle \quad \forall w\in \mathcal A^*,\label{eq:polynomial_add}\\
\langle \sigma\times \tau,w\rangle \triangleq \sum_{\substack{f,g\in \mathcal A^*:\\ f.g=w}}\langle \sigma,f\rangle  \langle \tau,g \rangle \quad \forall w\in \mathcal A^*.\label{eq:polynomial_prod}
\end{align}
We will use the usual symbols $+$ and $.$ in place of $\oplus$ and $\times$ in this work for convenience. The meaning of the operation would be clear depending on the operands. With these operations the set of polynomials form a non-commutative ring, and is denoted by $\mathbb Z\langle\mathcal A \rangle$, also called the free $\mathbb Z$-algebra on $\mathcal A$ in ring theory. Note that the addition and multiplication operations defined in \eqref{eq:polynomial_add} and \eqref{eq:polynomial_prod} are similar to the operations defined on commutative polynomials, except that the multiplication  operation under the summation in \eqref{eq:polynomial_prod} ($f.g=w$) is actually concatenation and is non-commutative.   The multiplication inside the summation in \eqref{eq:polynomial_prod} is multiplication in the real field and hence commutative. The multiplication defined in \eqref{eq:polynomial_prod} distributes over addition defined in \eqref{eq:polynomial_add}. Thus, a polynomial in $\mathbb Z\langle\mathcal A \rangle$ can be represented as a sum of monomials in $\mathcal A^*$ each with an associated coefficient in $\mathbb Z$, i.e., $\sigma=\sum\limits_{w\in \mathcal A^*} \langle\sigma,w \rangle w$. Define the \textit{degree} of a polynomial to be equal to the length of a longest sequence with a non-zero coefficient in the polynomial and the \textit{number of terms} of a polynomial as the number of sequences with non-zero coefficients in the polynomial. Note that a degree $d$ polynomial could have a number of terms upto $2^{d+1}-1$.
 
With this, the \textit{infiltration product} (in general, for two polynomials) is defined as follows:
\begin{align}
\forall f\in \mathcal{A}^*,& \quad f\uparrow e = e\uparrow f=f.\nonumber \\
\forall f,g\in \mathcal{A}^*&,\quad  \forall a,b\in \mathcal{A}, \nonumber\\
fa\uparrow gb=(f\uparrow gb)a&+(fa\uparrow g)b+\mathbbm{1}_{a=b}(f\uparrow g)a.\nonumber \\
\forall \sigma,\tau \in \mathbb{Z}\langle\mathcal{A} \rangle, \quad &\sigma\uparrow \tau=\sum_{f,g\in \mathcal{A}^*} \langle \sigma,f \rangle \langle \tau,g \rangle (f\uparrow g). \label{def:infiltforseq}
\end{align}

\subsection{Symbolwise posterior probabilities for the remnant channel}
\label{app:remnant_postprob}
Consider the remnant channel shown below, and let $Z=Z_1Z_2...Z_n$. Also let $Z_i \sim \text{Ber}(0.5)$. We aim to compute $\Pr(Z_i=1|\tilde  Y^{1}=y^1,\tilde Y^{2}=y^2,...,\tilde  Y^{t}=y^t)$. 
\begin{figure}[!h]
\centering
 \includegraphics[scale=0.4]{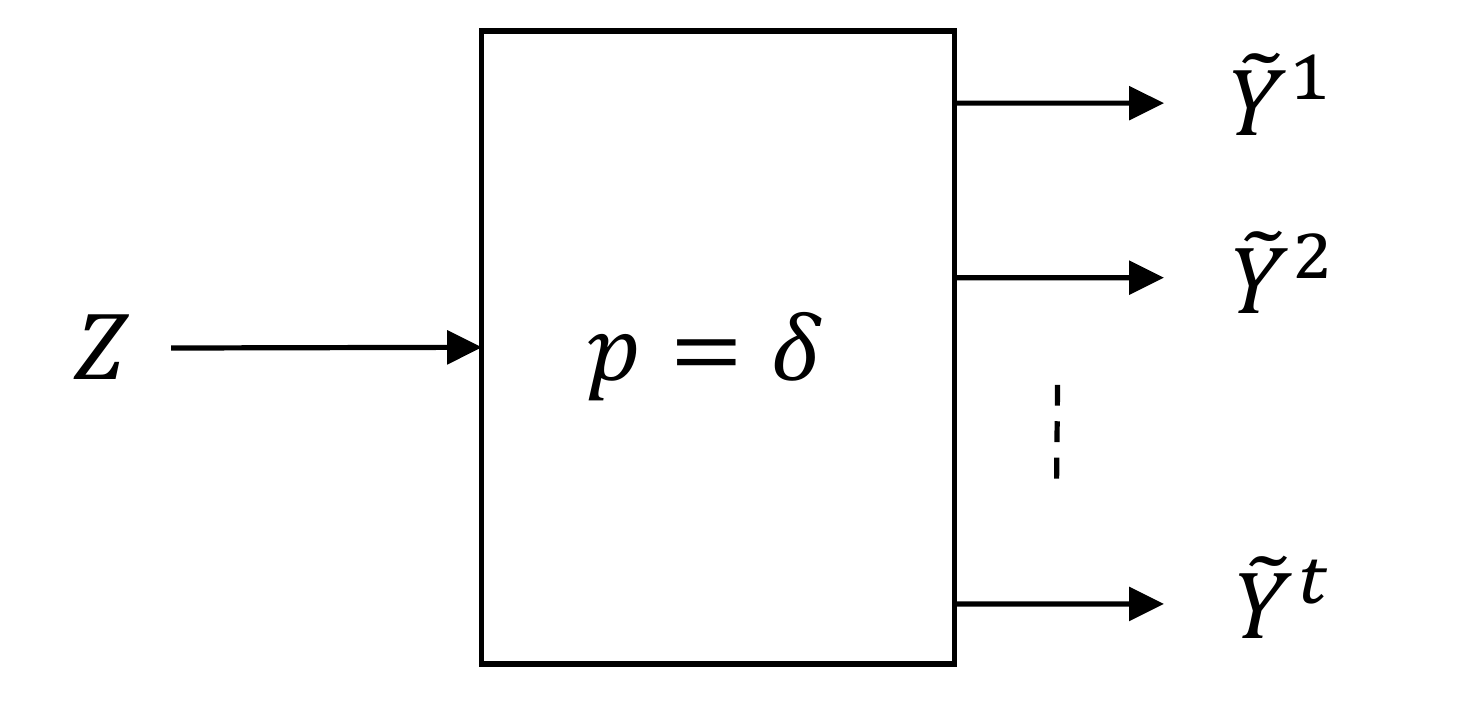}
 \caption{The remnant channel}
\end{figure}
From the definition of the infiltration product, the input-output relation for this channel can be derived to be:
\begin{align*}
\Pr(\tilde  Y^{1}=y^1,\tilde Y^{2}=y^2,...,\tilde  Y^{t}=y^t|Z) =\langle y^1\uparrow y^2 \uparrow...\uparrow y^t,Z \rangle \frac{(1-\delta)^{\sum|y^j|}\delta^{nt-\sum |y^j|}}{(1-\delta^t)^{n}}.
\end{align*}
Now, one could write the symbolwise posterior probabilities for $Z$ as:
\begin{align*}
\Pr(Z_i=1&|\tilde  Y^{1}=y^1,\tilde Y^{2}=y^2,...,\tilde  Y^{t}=y^t) = \sum_{\substack{z||z|=n,\\z_i=1}} \Pr(z|\tilde  Y^{1}=y^1,\tilde Y^{2}=y^2,...,\tilde  Y^{t}=y^t)\\
&{=} \frac{1}{2^n \Pr(\tilde  Y^{1}=y^1,\tilde Y^{2}=y^2,...,\tilde  Y^{t}=y^t)} \sum_{\substack{z||z|=n,\\z_i=1}} \Pr(\tilde  Y^{1}=y^1,\tilde Y^{2}=y^2,...,\tilde  Y^{t}=y^t|z)\\
&{=} \frac{{(1-\delta)^{\sum|y^{j}|}\delta^{nt-\sum |y^{j}|}}}{(1-\delta^t)^{n} 2^n \Pr(\tilde  Y^{1}=y^1,\tilde Y^{2}=y^2,...,\tilde  Y^{t}=y^t)} \sum_{\substack{z||z|=n,\\z_i=1}} \langle y^1 \uparrow y^2 \uparrow ... \uparrow y^t,z \rangle. \numberthis
\label{eq:remant_map_prob_1}
\end{align*}
A similar expression can be obtained for the case when $Z_i=0$ as 
\begin{align*}
\Pr(Z_i=0&|\tilde  Y^{1}=y^1,\tilde Y^{2}=y^2,...,\tilde  Y^{t}=y^t)\\
&{=} \frac{{(1-\delta)^{\sum|y^{j}|}\delta^{nt-\sum |y^{j}|}}}{(1-\delta^t)^{n} 2^n \Pr(\tilde  Y^{1}=y^1,\tilde Y^{2}=y^2,...,\tilde  Y^{t}=y^t)} \sum_{\substack{z||z|=n,\\z_i=0}} \langle y^1 \uparrow y^2 \uparrow ... \uparrow y^t,z \rangle. \numberthis 
\label{eq:remant_map_prob_0}
\end{align*}
We could further simplify \eqref{eq:remant_map_prob_1} and \eqref{eq:remant_map_prob_0} using the fact that the expressions in \eqref{eq:remant_map_prob_1} and \eqref{eq:remant_map_prob_0} must sum to 1, leading us to
\begin{align*}
\Pr(Z_i=1|\tilde  Y^{1}=y^1,\tilde Y^{2}=y^2,...,\tilde  Y^{t}=y^t) = \frac{ \sum\limits_{\substack{z||z|=n,\\z_i=1}} \langle y^1 \uparrow y^2 \uparrow ... \uparrow y^t,z \rangle}{\sum\limits_{\substack{z||z|=n}} \langle y^1 \uparrow y^2 \uparrow ... \uparrow y^t,z \rangle}.   \numberthis \label{eq:remnant_map_prob}
\end{align*}

We precisely describe the algorithm which computes the terms in \eqref{eq:remnant_map_prob} in section~\ref{sec:exactsmap}, by exploiting the edit graph interpretation of the infiltration product, but give a high level idea below. The complexity of such an algorithm is $O((2n)^t)$ which is equal to the number of edges in the edit graph. Note that for a fixed number of traces, this algorithm is polynomial in the blocklength as opposed to a naive approach of iterating through all the $n$-length sequences.

Recall that $\langle y^1 \uparrow y^2 \uparrow ... \uparrow y^t,z \rangle$ is the number of paths from origin to destination of the edit graph $\mathcal G(y^1,y^2,...,y^t)$ which correspond to $z$. Therefore, 
$\sum_{\substack{z||z|=n}} \langle y^1 \uparrow y^2 \uparrow ... \uparrow y^t,z \rangle$ is equal to the number of $n$-length paths in $\mathcal G(y^1,y^2,...,y^t)$ from the origin to the destination. Note that the edit graph has no cycles, so this quantity can be efficiently computed via the following dynamic program -- the number of $n$ length paths from the origin to a vertex $v$ is equal to the sum of the number of $n-1$ length paths from the origin to the in-neighbors of $v$. Such a procedure iterates over the vertex set of $\mathcal G(y^1,y^2,...,y^t)$ exactly once.

The numerator term $\sum_{\substack{z||z|=n\\z_i=1}}\langle y^1 \uparrow y^2 \uparrow ... \uparrow y^t,z \rangle$ can be interpreted in a similar way: it is equal to the number of $n$-length paths in $\mathcal G(y^1,y^2,...,y^t)$ from the origin to the destination such that the $i^{th}$ edge of the path corresponds to a `1'. The algorithm for this, therefore, follows a similar principle but has an extra step. For each vertex $v$, we compute
\begin{itemize}
\item the number of paths from the origin to $v$ of length $0,1,...,n$,
\item the number of paths from $v$ to the destination of length $0,1,...,n$.
\end{itemize}
Next we iterate over all edges in $\mathcal G(y^1,y^2,...,y^t)$ corresponding to a `1' and accumulate the number of $n$ length paths which have this particular edge as its $i^{th}$ edge. Thus, this algorithm iterates over the vertex set twice and the  edge set of $\mathcal G(y^1,y^2,...,y^t)$ once.

\subsection{A heuristic for ML optimization with a single trace.}

{The proof of Theorem~\ref{thm:ML_relaxation} inspires a heuristic for sequence reconstruction (see Alg.~\ref{alg:cood_switch}):
\begin{itemize}
\item Start from a given point $p = (p_1,...,p_n) \in [0,1]^n$.
\item One round of iteration is defined as follows: fix a traversal order for the indices $\{1,2,...,n\}$. Traverse through the indices $i$ in order and make $p_i$ either 0 or 1 depending on whether $\mathbf F(p^{(i\rightarrow 0)},y)$ or $\mathbf F(p^{(i\rightarrow 1)},y)$ is larger. This ensures that $\mathbf F(p,y)$ never decreases.
\item At the end of the round, check if the resultant $p$ was already obtained at the end of a previous round: if so, end the algorithm (to prevent it from going into an endless cycle). Otherwise, start a new round from the resultant $p$.
\end{itemize}
The resultant $p$ at the end of a round  is a lattice point since we make each $p_i$ to be 0 or 1. Therefore, the algorithm will end after a finite number of steps; in the worst case it will iterate through all $2^n$ sequences, although in practice we observe that it ends in 4-5 rounds (tested up to a blocklength of 100). We also note that the complexity of each round is $O(n^3)$ since it iterates through $n$ coordinates and for each coordinate computes $\mathbf F(\cdot)$, which is $O(n^2)$.}

\begin{algorithm}[t!]
\caption{Coordinate switch ML heuristic}\label{alg:cood_switch}
\begin{algorithmic}[1]
\item Input: Blocklength $n$, Trace {$Y=y$}, Initial point $p = (p_1,p_2,...,p_n)$ \\ Outputs: Estimated sequence $\hat X$
\State Initialize visited set $\mathcal V = \emptyset$
\While {True} 
\State Compute $\mathcal F_i = |\mathbf F(p^{(i\rightarrow 1)},y)- \mathbf F(p^{(i\rightarrow 0)},y)|\ \forall\ i$ and let $\mathcal F = (\mathcal F_1,\mathcal F_2,...,\mathcal F_n)$.
\State Define the ordered list $\mathcal S =$ \texttt{argsort}$(\mathcal F)$ where \texttt{argsort}$(\mathcal F)$ returns the index set $[n]$ sorted by descending order of $\mathcal F$, i.e., $\mathcal F_{\mathcal S_1}\geq \mathcal F_{\mathcal S_2}\geq ... \geq \mathcal F_{\mathcal S_n}$.
\For {$i \in \mathcal S$ (ordered traversal)}
\If {$\mathbf F(p^{(i\rightarrow 1)},y)- \mathbf F(p^{(i\rightarrow 0)},y) \geq 0$}
\State update $p \leftarrow p^{(i\rightarrow 1)}$
\Else
\State update $p \leftarrow p^{(i\rightarrow 0)}$
\EndIf
\EndFor
\If {$p \in \mathcal{V}$} break
\EndIf
\State $\mathcal V = \mathcal V \cup \{p\}$
\EndWhile
\State \textbf{return} $\hat X = p$
\end{algorithmic}
\end{algorithm} 


A natural question  is whether it makes a difference if Alg.~\ref{alg:cood_switch} starts from an interior point ($p = (p_1,...,p_n) \in [0,1]^n$ where $\exists\ p_i \in (0,1)$) as compared to starting from a lattice point (for instance, we could start from $p = (y,0,...,0) \in \{0,1\}^n$) which is the $n$-length sequence obtained via appending $y$ with zeros. It turns out that starting from an interior point results in better accuracy on both Hamming and edit error rate metrics, thus supporting the usefulness of our ML relaxation result. 

In Fig.~\ref{fig:singletrace}, we compare the performance of Coordinate switch heuristic with the other trace reconstruction heuristics in Section~\ref{sec:Numerics}.  We see that the coordinate switch with interior point initialization performs very similar to the true ML sequence (obtained via exhaustive search), in terms of both the Hamming error rate as well as the edit error rate. This intuitively supports the idea that this is a good heuristic for the ML optimization problem. However, at this point the heuristic is applicable for reconstruction using just a single trace and it is unclear on how to extend it to multiple traces.

\begin{figure}[!h]
\centering
\includegraphics[scale=0.5]{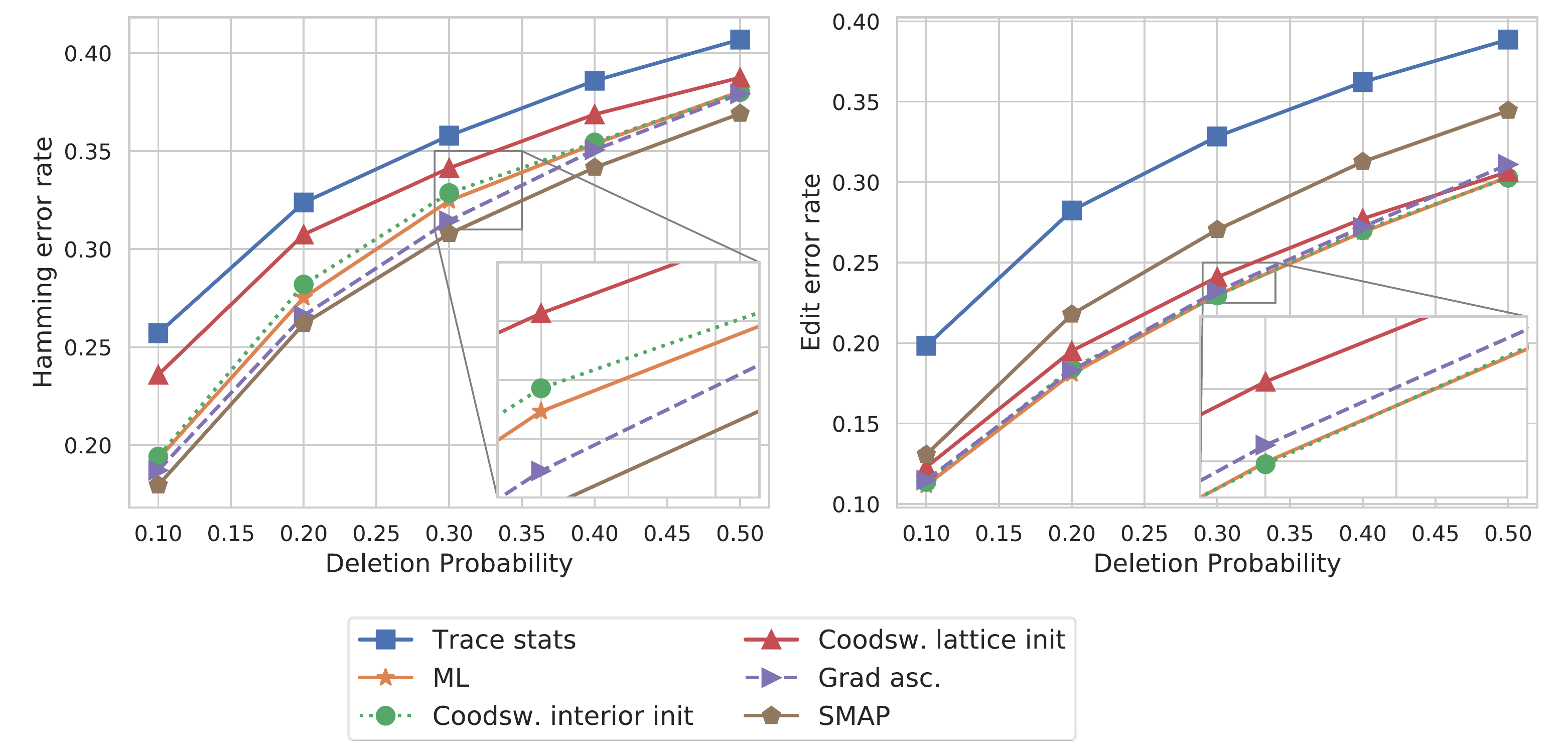}
\caption{Numerics for reconstruction from a single trace for a blocklength $n=20$. This plot compares the performance of coordinate switch heuristic (abbreviated ``Coodsw. interior init.'' and ``Coodsw. lattice init.'') with other trace reconstruction algorithms from Section~\ref{sec:Numerics}. ``ML'' refers to the true ML sequence obtained via an exhaustive search on all 20 length binary sequences. The interior point initialization initializes $p=(0.5,0.5,...,0.5)$ while the lattice point initialization appends the trace $y$ with zeros to obtain an $n$-length vector $p=(y,0,...,0)$.}
\label{fig:singletrace}
\end{figure}

\subsection{Symbolwise MAP as the minimizer of Hamming error rate}
\label{app:smap_hamming}

Symbolwise MAP  is an optimal estimator for minimizing the Hamming error rate for any channel, regardless of whether it is memoryless or not. This fact can be seen from the following argument: Consider a fixed observation $y$ (note that $y$ here can also be a collection of multiple observations, our arguments which follow remain unchanged) and that we aim to estimate a binary input sequence $X$; let the estimate of the input be $\hat X(y)$. Note that the estimate is a function of observation $y$ alone. Now the Hamming error rate of any estimator given $y$ is the expectation (over all inputs) of number of symbol mismatches divided by the blocklength, i.e.,
\begin{align*}
\frac{1}{n}\E \left[ \sum_{i=1}^n \mathbbm{1} \{X_i \neq \hat X_i(y)\} \Big| Y=y \right ] 
								&= \frac{1}{n} \sum_{i=1}^n \E \left [ \mathbbm{1}\{X_i \neq \hat X_i(y)\} \Big| Y=y  \right ]\\
& = \frac{1}{n}  \sum_{i=1}^n \Pr\left ( X_i \neq \hat X_i(y)\Big| Y=y \right ) \\
& = \frac{1}{n} \sum_{i=1}^n \Bigg( \Pr(X_i = 0|Y=y)\Pr(\hat X_i(y) = 1|X_i = 0,Y=y) \\ &  \hspace{1cm}+  
 \Pr(X_i = 1|Y=y)\Pr(\hat X_i(y) = 0|X_i = 1,Y=y) \Bigg ).
\end{align*}
But, $\hat X_i$ is a function of only $y$ and hence is conditionally independent of $X_i$ given $y$, which implies the following:
\begin{align*}
\frac{1}{n}\E \left[ \sum_{i=1}^n \mathbbm{1} \{X_i \neq \hat X_i(y)\} \Big| Y=y \right ] 
& = \frac{1}{n} \sum_{i=1}^n \Bigg( \Pr(X_i = 0|Y=y)\Pr(\hat X_i(y) = 1|Y=y) \\ &  \hspace{2cm}+  
 \Pr(X_i = 1|Y=y)\Pr(\hat X_i(y) = 0|Y=y)\Bigg ).
\end{align*}
To simplify notation, let the posterior probabilities be $q_i(y) \triangleq \Pr(X_i = 1|Y=y)$ and let $\alpha_i(y) \triangleq \Pr(\hat X_i(y) = 1|Y=y)$. Note that  $q_i(y)$ is a property of the channel and is fixed given $y$, while $\alpha_i(y)$ depends on the design of our estimator. With this, the above expression can be re-written as 
$$\frac{1}{n}\E \left[ \sum_{i=1}^n \mathbbm{1} \{X_i \neq \hat X_i(y)\} \Big| Y=y \right ] = \frac{1}{n} \sum_{i=1}^n \Bigg( (1-q_i(y)) \alpha_i(y) +    q_i(y) (1-\alpha_i(y))\Bigg ).$$ The optimal assignment of $\alpha_i(y)$ to minimize this expression is $\alpha_i(y) = 1$ if $q_i(y) \geq 0.5$ and $\alpha_i(y) = 0$ otherwise, which coincides with the symbolwise MAP estimate. This proves the optimality of symbolwise MAP for minimizing the Hamming error rate given any observation $y$, for any channel. 